\documentclass[onecolumn]{IEEEtran}

\usepackage{amsmath,amssymb,amsfonts}
\usepackage[ruled,vlined,linesnumbered]{algorithm2e}
\SetKwInOut{Init}{Initialize}
\SetKwInput{KwInput}{Input}
\SetKwInput{KwOutput}{Output}
\SetKwRepeat{DoWhile}{do}{while}
\SetKwBlock{Legend}{Legend}{end legend}
\usepackage{graphicx, float}
\usepackage{textcomp}
\usepackage{xcolor}
\usepackage{bbm}
\usepackage{comment}
\usepackage{enumitem}
\usepackage{tikz}
\usetikzlibrary{positioning, arrows.meta}
\usepackage{tabularx}
\usepackage{dsfont}
\usepackage{multirow}
\usepackage{amsthm}
\usepackage{subcaption}
\usepackage{url}

\newtheorem{lemma}{Lemma}

\newtheorem{remark}{Remark}

\usepackage{tablefootnote}
\usepackage{booktabs}

\begin{document}

\title{Techno-Economic Analysis of Shared Mobile Storage for Demand Charge Reduction}

\author{B Hari Kiran Reddy, \textit{Student Member, IEEE}, Ge Chen, \textit{Member, IEEE}, and Junjie Qin, \textit{Member, IEEE}%
\thanks{B. H. K. Reddy and J. Qin are with the Elmore Family School of Electrical and Computer Engineering, Purdue University, West Lafayette, IN, USA. G. Chen is with the School of Advanced Engineering, Great Bay University, Dongguan, China. 
(Emails: \texttt{reddy83@purdue.edu}; \texttt{gechen@gbu.edu.cn}; \texttt{jq@purdue.edu})}}

\maketitle

\begin{abstract}
This paper investigates the techno-economic viability of shared electric vehicle (EV) fleets for demand charge reduction under practical logistical and operational constraints. Unlike idealized models that overlook transit overheads, we propose a high-fidelity fleet management framework that explicitly accounts for the spatio-temporal coupling of energy consumption, labor costs for EV drivers, and battery degradation. We formulate the dispatch problem as a mixed-integer linear program (MILP) that jointly minimizes demand charges and total cost of ownership. To address the computational complexity arising from path-dependent constraints, we develop a marginal-value-based heuristic algorithm that achieves high-quality feasible solutions with high computational efficiency. Using real-world data from San Francisco, our analysis reveals that a modest number of EVs can achieve significant demand charge savings, sufficient to recover the ownership and operational expenses. Our results also show how tariff structures, fleet size, and cost components influence overall profitability.
\end{abstract}


\vspace{-0.25cm}
\section{Introduction}

Electric vehicles (EVs) are widely recognized as an eco-friendly alternative to internal combustion engine vehicles, offering lower greenhouse gas emissions and reduced operating costs \cite{sanguesa2021review}. Some EVs like Nissan Leaf \cite{9203459} can both draw energy from the grid and supply it back to external loads when equipped with bidirectional chargers. Beyond mobility, these vehicles provide added value through services such as emergency backup power and demand response \cite{kazemtarghi2022optimal}.

A particularly promising avenue is the integration of EVs with commercial and industrial (C\&I) facilities to address demand charges. These tariff components, which are proportional to the maximum real power demand averaged over short intervals (e.g., 15 minutes), often constitute a significant portion of electricity costs for C\&I users. While stationary battery energy storage systems (BESS) are the conventional solution for peak shaving, they suffer from high upfront investment costs  and spatial rigidity, as an asset installed at one site cannot serve another~\cite{schmidt2017dynamic, luo2015overview}. In contrast, by leveraging the mobility of EVs, a single shared asset can serve multiple users across different locations during their respective peaks. This \textit{mobile storage} paradigm theoretically offers superior asset utilization and a lower investment barrier compared to stationary infrastructure.

However, transforming this theoretical promise into a viable business operation requires overcoming significant logistical hurdles often overlooked in existing literature. Emerging studies frequently rely on idealized assumptions that neglect the spatio-temporal coupling of energy consumption and the labor costs of EV drivers, a dominant yet often omitted operational factor. Furthermore, the complex trade-off between battery degradation and varying charging infrastructure capabilities is rarely integrated into a unified dispatch objective. Neglecting these granular factors may yield solutions that are technically optimal but economically inviable. This motivates us to study the following critical open question: \textit{How much net value can truly be realized once demand charge savings from EV services are weighed against charging costs, battery wear, labor expenses, commuting requirements, and charger investments?}

\subsection{Related literature}\label{sec:RL}
Existing literature on peak demand reduction spans stationary resources, EV scheduling, and shared mobile storage. We review this literature from the perspectives of modeling, solution methods, and techno-economic assessment.

\subsubsection{Modeling storage and EV coordination}
Stationary resources, including BESS, diesel generation, and PV+storage hybrids, have been studied for peak shaving~\cite{zahari2024integrating,li2018performance}, but cannot serve users across different locations. EV-based studies have established opportunities for grid support~\cite{jian2012regulated,alam2014controllable} and bill reduction through managed charging and vehicle-to-building (V2B) services~\cite{neubauer2015deployment,ghafoori2023electricity}, often within a given facility or parking site.

Shared-EV studies extend this scope to demand-charge reduction and uncertainty-aware service matching~\cite{qin2021mobile,chen2026monetizing,10542422}, while related work examines transportation-energy coordination, storage relocation, and fleet logistics~\cite{8755339,agwan2023mobile,9779353}. Mobility enables an EV to serve different users at their respective peaks, but couples service assignments with transit, charging, and battery availability. Jointly representing these requirements and the associated labor, transit-energy, and depreciation costs remains insufficiently explored for fleets dedicated to multi-user demand-charge reduction.

\subsubsection{Solution methods for EV scheduling}
EV scheduling methods include MILP~\cite{liu2024peak}, QP~\cite{ioakimidis2018peak}, stochastic~\cite{cardoso2013stochastic,sen2025online}, and robust optimization~\cite{zhao2025day} under operational constraints~\cite{mousaei2024advancing}. For shared fleets, the interdependence of service decisions increases scheduling complexity as the number of EVs and users grows. This creates a need for computationally tractable monthly scheduling methods, together with empirical assessments of solution quality and runtime for offline fleet planning.

\subsubsection{Techno-economic assessment}
Prior work evaluates stationary-resource economics~\cite{zahari2024integrating,li2018performance} and EV bill savings~\cite{alam2014controllable,neubauer2015deployment}. For the shared-fleet applications discussed above, demand-charge reduction alone does not establish fleet profitability: savings must offset charging, transit, labor, depreciation, and infrastructure costs. How these costs interact with fleet size, uniform or differentiated charging configurations, and tariff structures remains insufficiently characterized for shared fleets serving independent C\&I users.

\subsection{Contributions and paper organization} 
To address these gaps, we develop a high-fidelity fleet management framework that rigorously quantifies the economic value of EV fleets for demand charge reduction under practical operating conditions. Here, ``high-fidelity'' refers to the inclusion of operational and economic factors that materially affect the techno-economic feasibility of shared mobile energy storage, including transit energy consumption, travel-dependent dispatch coupling, labor cost, charging infrastructure cost, and battery degradation costs represented through approximate depreciation models.  By explicitly accounting for these practical considerations, which are often neglected in prior mobile energy storage studies, the proposed framework establishes a more realistic baseline for evaluating the commercial viability of shared mobile energy storage. The main contributions of this work are as follows:

\begin{enumerate}
    \item We formulate a mixed-integer linear program (MILP) that explicitly captures the spatio-temporal coupling of mobile storage. Unlike prior studies, our model integrates transit energy consumption, labor costs, and battery degradation into a unified dispatch objective, ensuring operational feasibility and accurate economic assessment.

    \item To tackle the computational complexity associated with user-EV dispatch decisions and the corresponding discharging, transit and charging requirements, we propose a marginal-value-based heuristic algorithm. It iteratively dispatches the least-used EVs to the most-valued services, achieving high-quality feasible solutions while significantly reducing computation time relative to the MILP formulation for offline fleet planning.

    \item Through a case study using real-world data from San Francisco, we identify labor costs as the dominant factor affecting the revenue. Furthermore, we propose and validate a \textit{tiered} charger deployment strategy (mixing DC fast and AC Level-2 chargers across  users), demonstrating that it yields superior net savings compared to uniform infrastructure configurations.
\end{enumerate}

The remaining parts are organized as follows. Section~\ref{sec:FRM} describes the centralized framework to evaluate the profitability of the business model. Section~\ref{sec:HA} introduces the proposed algorithm. Section~\ref{sec:CS} demonstrates the case study using the empirical data, and Section~\ref{sec:CON} concludes the paper and discusses future work.

\section{Problem Formulation}\label{sec:FRM}

We consider a centralized fleet management framework where a dedicated operator coordinates a fleet of EVs to provide peak-shaving services to a portfolio of C\&I facilities. Unlike stationary storage, these mobile assets are subject to complex spatio-temporal constraints: providing service at a specific location requires physically traversing the road network, which incurs both time delays and energy consumption. The operator's objective is to determine the optimal dispatch schedule, including EV travel decisions, charging, and discharging profiles, that minimizes the system-wide cost. This cost function explicitly accounts for the trade-off between demand charge savings, battery degradation, transit energy expenditure, and charging electricity costs. Figure~\ref{fig:Mdl} provides the schematic of this framework.

 \begin{figure}[h]
\hspace*{4cm}
\begin{tikzpicture}[
        platform/.style={rectangle, draw=black, thick, line width=1.5pt, minimum width=2cm, minimum height=1cm, align=center},
        ev/.style={rectangle, draw=black, thick, line width=1.5pt, minimum width=2cm, minimum height=1cm, align=center},
        cs/.style={rectangle, draw=black!70, thick, line width=1.5pt, minimum width=2cm, minimum height=1cm, align=center},
        cni/.style={rectangle, draw=black, thick,line width=1.5pt, minimum width=2cm, minimum height=1cm, align=center},
        legend/.style={rectangle, draw=black, line width=1.5pt, fill=white, thick, align=left, minimum width=6cm, minimum height=0.5cm,, font=\footnotesize}
    ]

    \node[platform] (platform) {Fleet Operator};
    \node[ev, below left = 1.8cm of platform] (ev) {EV Drivers};
    \node[cni, below right=1.8cm of platform] (cni) {C\&I Users};
    \draw[<-, thick, line width=1.5pt, >=Stealth, blue, text=blue]  (-4.25,-1.8) -- (-1.2,0.25)  node[midway, above left = 0cm, font = \footnotesize, align=center] {service intervals,\\ service location};
    \draw[<-, thick, line width=1.5pt, >=Stealth, orange, text=orange]  (-3.5,-1.8) -- (-1.2,-0.25)  node[midway, below right = 0cm, font = \footnotesize, align=center] {Payment};
    \draw[<-, thick, line width=1.5pt, >=Stealth, blue, text=blue]  (1.2,0.25) -- (4.25,-1.8)  node[midway, above right = 0cm, font = \footnotesize, align=center] {load profile,\\user location};
    \draw[<-, thick, line width=1.5pt, >=Stealth, orange, text=orange]  (1.2,-0.25) -- (3.5,-1.8)  node[midway, below left = 0cm, font = \footnotesize, align=center] {Payment};

    \draw[->, thick, line width=1.5pt, >=Stealth, color=black] (-2.3,-2.3) --  (2.3,-2.3) node[midway, above = 0cm, font = \footnotesize, align=center] {Energy during service intervals};
    \node[legend, line width=1.5pt, below=2.5cm of platform] (legend) {
        \textbf{Legend:} \quad
        \tikz \draw[thick, ->, line width=1.5pt, >=Stealth, color=blue] (0,0) -- (0.7,0); Information
        \tikz \draw[thick, ->, line width=1.5pt, >=Stealth, color=black] (0,0) -- (0.7,0); Energy
        \tikz \draw[thick, ->, line width=1.5pt, >=Stealth, color=orange] (0,0) -- (0.7,0); Payment
    };
   
\end{tikzpicture}
\caption{Schematic of centralized fleet management framework.}
\label{fig:Mdl}
\end{figure}
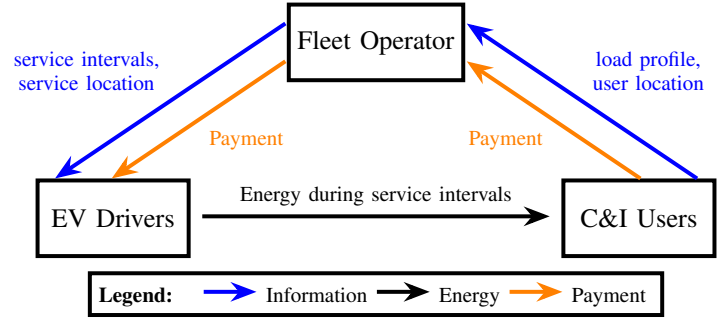

For example, consider a simple setting with two users and two EVs which are coordinated through a centrally managed charging station. The operator determines whether each EV should remain at the charging station for recharging or be dispatched to one of the users based on the resulting demand charge savings and the associated operational costs. Once dispatched, each EV must first spend some time in transit before providing service at the user location. Depending on the temporal variation of user demand and the battery state-of-charge levels, both EVs may simultaneously serve different users, or multiple EVs may sequentially serve the same user across different intervals if prolonged peak-shaving support is required. After completing the service, an EV may either be dispatched to another user if additional profitable service opportunities exist or return to the charging station for recharging. Consequently, dispatch decisions, charging schedules, transit requirements, and service intervals become strongly coupled across both users and EVs over time. An illustrative workflow of this two-user, two-EV example is shown in Figure~\ref{fig:sim_work_2u_2ev}.

\begin{figure}[htbp]
    \centering
    \includegraphics[width=\linewidth]{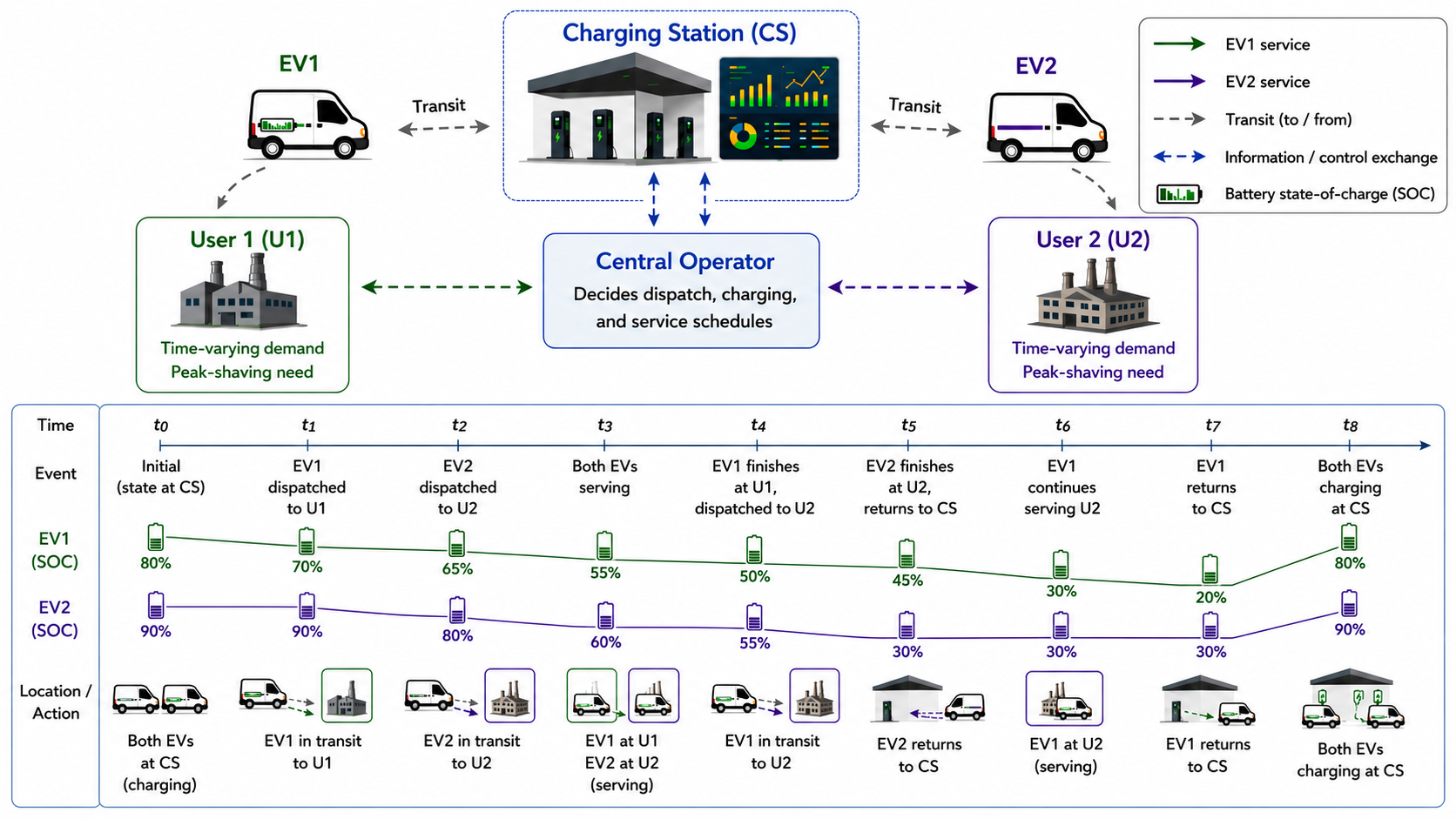}
    \caption{Illustrative workflow for the two-user, two-EV example.}
    \label{fig:sim_work_2u_2ev}
\end{figure}

We consider $I$ users and $J$ EVs participating in this business model. 
As our primary goal is to assess the viability of the business model at the planning horizon, we adopt the following setting typical in planning studies.
In particular, each user $i \in \mathcal{I} := \{1,\cdots,I\}$ provides monthly energy consumption profiles obtained from historical data or forecasting tools, while each EV $j \in \mathcal{J} := \{1,\cdots,J\}$ is dedicated to delivering electricity services over the study horizon. 

Under these model choices, we formulate the optimization problem for the fleet operator to dispatch the EVs optimally. The objective is to minimize the users' demand charges and EVs' operational cost for a month. The decision variables are the dispatch matrices  $\left(\mathbf{M}_{j} \right)_{j\in \mathcal{J}}$. Here, $\mathbf{M}_{j} \in \{0,1\}^{T \times I}$ represents EV $j$'s dispatch to users' locations during the set of intervals $\mathcal{T} = \{1,\cdots,T\}$ where $T$ is the total number of intervals of the month.  An entry $ M_{j,t,i} = 1 $ in $\mathbf{M}_{j}$ represents that EV~$j$ is either in transit to user~$i$'s location or is at user~$i $'s location during the  interval~$t$. Conversely, $M_{j,t,i} = 0$ suggests EV $j$'s absence from user~$i$’s location during interval $t$. Here, we consider that the transit between any two locations takes a nonzero duration that is shorter than one interval, and that corresponding interval is reserved exclusively for transit\footnote{This is reasonable given the close geographic proximity of the users, which leads to short travel times relative to the modeling interval length.}.

For notational simplicity, let $\mathbf{M} := \left(\mathbf{M}_{j}\right)_{j \in \mathcal{J}}$ denote the collection of dispatch matrices. Then, we can express the optimization problem for the fleet operator in determining the dispatch of EVs as:
\begin{subequations}\label{prb:plt}
\begin{align}
 \min_{  \mathbf{M} } \quad & \sum_{i \in \mathcal{I}}  \mathrm{DC}_{i} \left(  \mathbf{M}  \right) +  \mathrm{OC}\left( \mathbf{M} \right)  \\
\textrm{s.t.} \quad & \mathbf{M}_{j} \in \{0,1\}^{T \times I}, \quad \qquad j \in \mathcal{J}, \label{dom_M}\\
& M_{j,1,i} = 0, \qquad \qquad \quad \,\,\, j\in \mathcal{J}, i \in \mathcal{I},\label{alloc_mat:const_0}\\
& \mathbf{M}_{j}\mathbf{1} \leq \mathbf{1}, \qquad \qquad \qquad  j \in \mathcal{J}, \label{alloc_mat:const_1}\\
&\sum_{ j \in \mathcal{J}} M_{j,t-1,i} \, M_{j,t,i} \leq 1, \, t \in \mathcal{T}\setminus \{1\}, \, i \in \mathcal{I}, \label{alloc_mat:const_2}
\end{align}
\end{subequations}
where $\mathrm{DC}_{i}$ represents the demand charge of user $i$ and $\mathrm{OC}$ is a comprehensive term that captures the physical costs of mobility, including the battery depreciation cost due to cycling and the grid electricity cost for recharging the fleet.
$\mathrm{DC}_{i}$ and $\mathrm{OC}$ can be formulated as a function of dispatch matrices, which will be introduced in detail in Subsections~\ref{ss:ci_user} and \ref{EV_op} respectively. Constraint~\eqref{alloc_mat:const_0} specifies the initial fleet configuration, with no EV dispatched to any user during the first interval. Constraint~\eqref{alloc_mat:const_1} ensures that each EV is dispatched to at most one user during any interval.

The last constraint reflects the physical limitation that each user is equipped with a single bidirectional charger and therefore at most one EV can actively provide service during any interval~$t$. A straightforward formulation, $\sum_{j \in \mathcal{J}} M_{j,t,i} \leq 1$,
is overly restrictive because \(M_{j,t,i}=1\) indicates that EV \(j\) is either traveling toward or physically present at user \(i\). Consequently, this formulation would incorrectly prevent a second EV from traveling to the same user location while another EV is actively providing service.\footnote{This distinction becomes important when peak-shaving service is required over multiple consecutive intervals. In such cases, a single EV may not possess sufficient battery energy to sustain uninterrupted service, requiring a second EV to begin transit toward the user before the first EV departs.}

To illustrate this behavior, consider an example with two EVs serving user \(i\), where peak-shaving service is required during intervals \(t, t+1, t+2,\) and \(t+3\) as shown in Figure~\ref{fig:ts_ill_ex}. Suppose EV~1 begins traveling toward the user during interval \(t-1\) and subsequently provides service during intervals \(t, t+1,\) and \(t+2\). Since one transit interval is required before service can begin,
\[
M_{1,t-1,i}=M_{1,t,i}=M_{1,t+1,i}=M_{1,t+2,i}=1.
\]
Now assume EV~1 does not have sufficient remaining battery energy to continue service during interval \(t+3\), forcing $M_{1,t+3,i}=0$.
To maintain uninterrupted service, EV~2 must begin serving the user during interval \(t+3\). Since one transit interval is again required before service can begin, EV~2 must already be traveling toward the user during interval \(t+2\), resulting in $M_{2,t+2,i}=M_{2,t+3,i}=1$.

As a consequence, $M_{1,t+2,i}+M_{2,t+2,i}=2$,
even though only EV~1 is actively providing service during interval \(t+2\). Here, EV~1 is serving the user, while EV~2 is in transit toward the same location in preparation for the service handoff. Therefore, directly enforcing
$\sum_{j \in \mathcal{J}} M_{j,t,i} \leq 1$
would incorrectly eliminate this physically feasible transition. Instead, the bilinear constraint~\eqref{alloc_mat:const_2} ensures that at most one EV can maintain continuous service assignment across consecutive intervals \((t-1)\) and \(t\), while still permitting transit-related overlap required for uninterrupted service continuity. These consecutive assignments arise implicitly from the structure of \(\mathrm{DC}_{i}\), discussed in Subsection~\ref{ss:ci_user}. 

\begin{figure}[htbp]
    \centering
    \includegraphics[width=\linewidth]{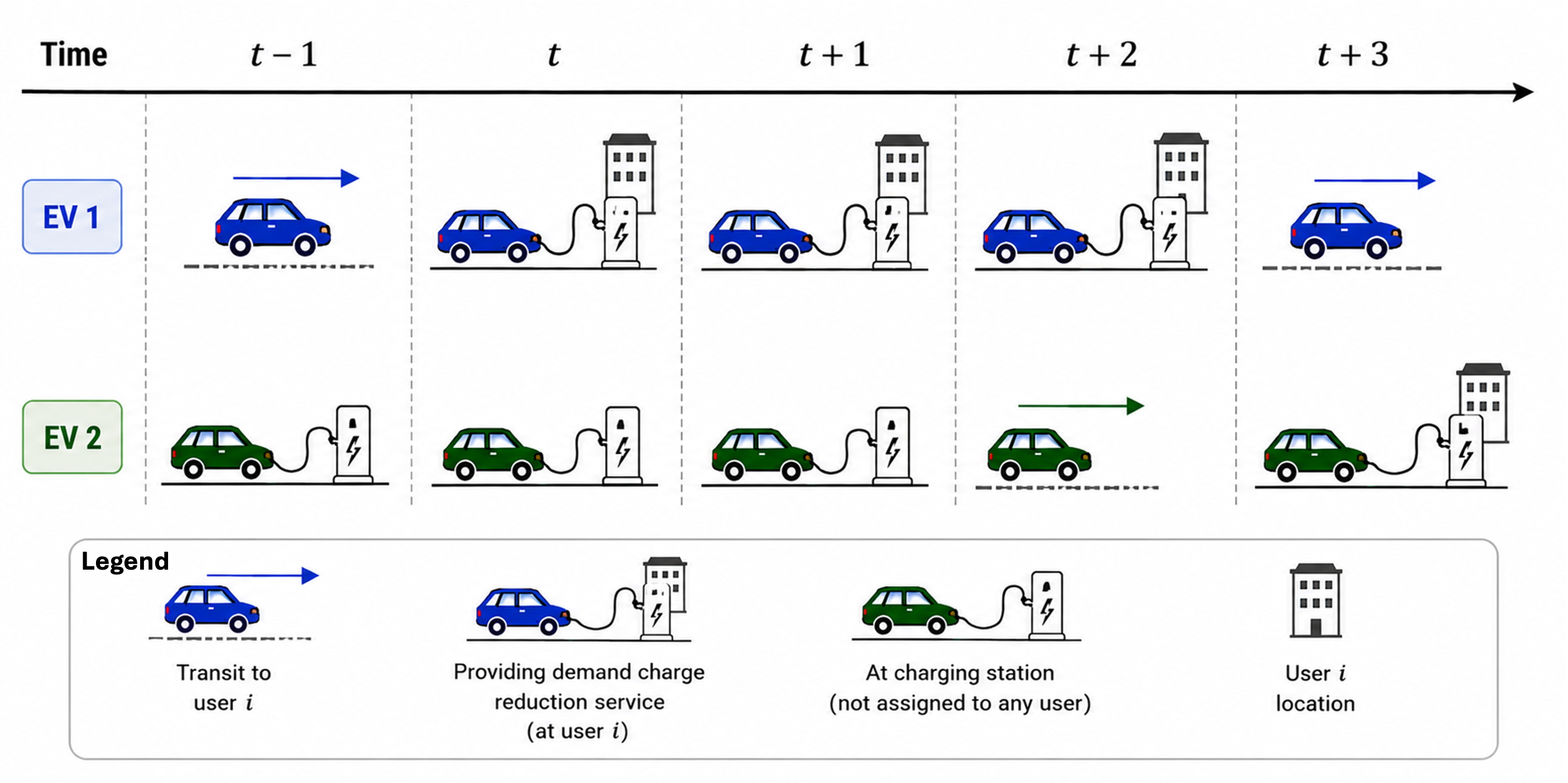}
    \caption{Illustrative timeline for a two-EV setting.}
    \label{fig:ts_ill_ex}
\end{figure}

\subsection{C\&I user's problem}\label{ss:ci_user}
This subproblem aims to minimize a user's demand charge (i.e., $\mathrm{DC}_{i}$ in problem~\eqref{prb:plt}) by optimally determining the service intervals given the dispatch matrices.
Dispatch matrices, by themselves, do not dictate the service intervals. They only establish the possibility of service.
For example, an EV could be dispatched to a user's location and serve the user in the current interval but the operator may find that providing service in the next interval does not lower the demand charge. In such cases, the operator may ask the EV driver to refrain from providing the service to the user in next interval and instead wait at the user’s location, anticipating a near-future interval when providing service does contribute to demand charge reduction.

For user $i$, let $\pmb{\ell}_{i} \in \mathbb{R}^{T}$ (in kWh) denote the vector of energy consumption measurements prior to the provision of peak-shaving services in a month. Many utilities like PG\&E compute demand charges across multiple time-of-use (TOU) periods (each of which containing a set of intervals) rather than a single peak, by identifying the maximum demand within each period and aggregating the corresponding charges. The TOU periods contributing to demand charges may differ across users. For user $i$, let $\Theta_{i}$ represent the set of TOU periods (e.g. peak, part-peak, off-peak) which contribute towards the demand charges, and for each TOU period $\theta \in \Theta_{i}$, let $\mathcal{T}^{\theta} \subseteq \mathcal{T}$ denote the set of intervals belonging to that period. The peak demand in period~$\theta$ is then defined as $ \max_{t \in \mathcal{T}^{\theta}} \left(\ell_{i,t} \right)/ \Delta$ (in kW) where $\Delta$ is the duration of  intervals in hours. With demand rates $\pi_{i}^{\theta}$ (in \$/kW), the total demand charge without the services can thus be obtained as $ \sum_{\theta \in \Theta_{i}} \tilde{\pi}_{i}^{\theta} \max_{t\in \mathcal{T}^{\theta}} \ell_{i,t} $ where $ \tilde{\pi}^{\theta}_{i} :=  \pi^{\theta}_{i} / \Delta$.

To determine the service intervals, we introduce a variable $\mathbf{u}_{i} \in \{0,1\}^{T}$ where $u_{i,t} = 1$ indicates that an EV provides service during  interval $t$, and $u_{i,t} = 0$ indicates no service during that interval. Given the dispatch matrices $\mathbf{M}$, the optimization problem for user $i$ to minimize the demand charge post receiving the services can be given as:
\begin{subequations}\label{prb:user}
\begin{align}
    \mathrm{DC}_{i}\left( \mathbf{M} \right)   = &  \min_{\mathbf{u}_{i}} \quad \sum_{\theta \in \Theta_{i}} \tilde{\pi}^{\theta}_{i} \max_{t\in \mathcal{T}^{\theta}} \left( \ell_{i,t} - \overline{p}_{i}\Delta u_{i,t} \right)\label{user:obj_fun}\\
& \textrm{ s.t. } \quad \mathbf{u}_{i} \in \{0,1\}^{T}, \label{dom_u}\\
& \qquad \quad u_{i,1} = 0, \label{no_service_first_int}\\
 & \qquad \quad \mathbf{0} \leq \overline{p}_{i}\Delta \, \mathbf{u}_{i} \leq \pmb{\ell}_{i} \label{no_back},\\
& \qquad \quad u_{i,t} \leq \sum_{ j \in \mathcal{J}}  M_{j,t-1,i} \, M_{j,t,i} , \, t \in \mathcal{T} \setminus\{1\}\label{user:request},
 \end{align}
\end{subequations}
where $\overline{p}_{i}$ is bidirectional charger capacity (in kW) installed at the user's location. Constraint~\eqref{no_service_first_int} precludes user service during the first interval, consistent with initial fleet configuration in Constraint~\eqref{alloc_mat:const_0}.  Constraint~\eqref{no_back} ensures that energy is not returned to the grid as some utility companies may impose severe penalties. Constraint~\eqref{user:request} ensures that  interval~$t$ can be a service interval only if an EV is dispatched to the user during both $(t-1)$ and $t$ intervals as transit cost one interval.

\subsection{EVs' operational problem}\label{EV_op}

This subsection aims to determine the operational cost of the EVs (i.e., $\mathrm{OC}$ in problem \eqref{prb:plt}) given the dispatch matrices.  This cost consists of two parts: (a) depreciation cost, denoted as $\mathrm{Dep}$, and (b) charging cost, denoted as $\mathrm{CC}$. Accordingly, the operational cost can be given as $\mathrm{OC} = \mathrm{Dep} + \mathrm{CC}$. Next, we introduce these two parts in detail.

\subsubsection{Depreciation cost}\label{ss:ev_charging}

We account for the depreciation cost of each EV $j \in \mathcal{J}$ using a usage-based approximation\footnote{This approximation monetizes battery usage at the fleet-planning level and is not intended to represent detailed cell-level electrochemical degradation or nonlinear battery aging behavior.}. Specifically, a depreciation rate $\gamma_{j}$ (in \$/kWh), which serves as an economic proxy for battery wear, is multiplied by the energy used for providing electricity services $\mathbf{e}^{\mathrm{dis}}_{j} \in \mathbb{R}^{T}$ and the energy used for transit $\mathbf{e}^{\mathrm{tr}}_{j} \in \mathbb{R}^{T}$. Therefore, the total depreciation cost of all EVs is:
\begin{align}
\label{Dep_cost}
 \mathrm{Dep}\left( \mathbf{M} \right) = \sum_{ j \in \mathcal{J}} \gamma_{j} \mathbf{1}^{\intercal} \left(\mathbf{e}^{\mathrm{dis}}_{j}  + \mathbf{e}^{\mathrm{tr}}_{j}  \right).
\end{align}
For EV $j$, the value of $\mathbf{e}_{j}^{\mathrm{dis}}$ is calculated by:
 \begin{align}\label{ev:discharge}
 e^{\mathrm{dis}}_{j,t} =  \sum_{i \in \mathcal{I}}  \overline{p}_{i} \Delta \, u_{i,t}^\star \, M_{j,t-1,i} \, M_{j,t,i}
\end{align} for $t \in \mathcal{T} \setminus \{1\}$ where $u_{i,t}^\star$ corresponds to the service intervals of user $i$ for  interval $t$, obtained as the solution to the optimization problem \eqref{prb:user} with $e^{\mathrm{dis}}_{j,1} = 0$. In other words, $e^{\mathrm{dis}}_{j,t} = \overline{p}_{i}\Delta$ if EV~$j$ is dispatched to user $i$ during both intervals $(t-1)$ and $t$, and interval $t$ is a service interval for user $i$. If it is not true for all $I$ users, then $e^{\mathrm{dis}}_{j,t} = 0$. The value of $\mathbf{e}_{j}^{\mathrm{tr}}$ is calculated as the ratio of the distance traveled to the EV's range efficiency $\delta_{j}^{\mathrm{tr}}$ (in miles/kWh), which is given as:
\begin{align}
 e_{j,t}^{\mathrm{tr}} = \frac{1}{\delta_{j}^{\mathrm{tr}}} \big\| (x_{j,t}, y_{j,t})  - (x_{j,t-1}, y_{j,t-1}) \big\|_1
\label{ev:transit_ener}
\end{align}
for $t \in \mathcal{T} \setminus \{1\}$ with $e_{j,1}^{\mathrm{tr}} = 0$. Here, $\big\| \cdot \big\|_1 $ denotes the 1-norm, which is commonly used to approximate driving distance in urban context~\cite{krause1973taxicab}. Also, $\left(\mathbf{x}_{j}, \mathbf{y}_{j} \right) \in \mathbb{R}^{T} \times \mathbb{R}^{T} $ represents EV $j$'s \textit{x-y} coordinates over time. These coordinates depend on the dispatch matrices as:
\begin{subequations}
\begin{align}
 \mathbf{x}_{j}  =  x_{0} \left(  \mathbf{1} - \mathbf{M}_{j} \mathbf{1}\right)  +  \mathbf{M}_{j}  \mathbf{x}^{\mathrm{C\&I}} ,  \label{ev:x}\\
 \mathbf{y}_{j}  =  y_{0} \left(  \mathbf{1} - \mathbf{M}_{j} \mathbf{1}\right)  +  \mathbf{M}_{j} \mathbf{y}^{\mathrm{C\&I}}, \label{ev:y}
 \end{align}
\end{subequations}
where $ x_{0} $ and $y_{0}$ are the respective \textit{x-y} coordinates of a centrally managed charging station which all the EV recharge and the vectors $\mathbf{x}^{\mathrm{C\&I}} \in \mathbb{R}^{I}$ and $\mathbf{y}^{\mathrm{C\&I}} \in \mathbb{R}^{I}$ represent the coordinates of the C\&I users.

\subsubsection{EVs' charging cost}\label{ss:ev_charging}  We account for the EVs' charging cost by considering both the energy charge and the demand charge. This paper assumes that EVs only charge at the charging station. Poor management of EV charging scheduling may lead to high peak charging demands, resulting in a high demand charge. As a result, the EV charging schedules should also be optimized.

Given the dispatch matrices $\mathbf{M}$, we can optimize the charging schedules of each EV $j \in \mathcal{J}$, denoted as $\mathbf{e}_{j}^{\mathrm{ch}}~\in~\mathbb{R}^{T}$, and the corresponding state of charge of its battery, denoted as $\mathbf{b}_{j} \in \mathbb{R}^{T}$, such that the charging cost is minimized as follows:
\begin{subequations} \label{ev:charging_cost}
\begin{align}
 \mathrm{CC}\left( \mathbf{M} \right) = & 
 \min_{\left( \mathbf{e}_{j}^{\mathrm{ch}}, \mathbf{b}_{j}\right)_{ j \in \mathcal{J}}} \quad \pmb{\pi}_{0}^{\intercal}  \sum_{ j \in \mathcal{J}}\mathbf{e}_{j}^{\mathrm{ch}} + \sum_{\theta \in \mathrm{\Theta}_{0}}\tilde{\pi}^{\theta}_{0} \max_{t \in \mathcal{T}^{\theta}} \left(\sum_{ j \in \mathcal{J}}e_{j,t}^{\mathrm{ch}} \right) \\
& \quad \textrm{ s.t. } \qquad 0 \leq e_{j,t}^{\mathrm{ch}} \leq  \left(m_{j,t-1}\,m_{j,t}\right)\overline{p}_{0}\Delta , \, t \in \mathcal{T} \setminus \{1\}, \, j \in \mathcal{J}, \label{ev:char}\\
& \qquad \qquad \quad e^{\mathrm{ch}}_{j,1} = 0, \,\,  j \in \mathcal{J}, \label{eq:char_first_int}\\
&  \qquad \qquad \quad b_{j,t} = b_{j,t-1} + \eta_{j}^{\mathrm{ch}} e_{j,t}^{\textrm{ch}} -  \left(1/\eta_{j}^{\mathrm{dis}}\right) \left( e^{\mathrm{dis}}_{j,t}  + e^{\mathrm{tr}}_{j,t} \right), t \in \mathcal{T} \setminus \{1\}, \, j \in \mathcal{J}, \label{eq:batt_dyn}\\
& \qquad \qquad  \quad \underline{b}_{j}\mathbf{1} \leq \mathbf{b}_{j}\leq \overline{b}_{j} \mathbf{1}, \, j \in \mathcal{J}, \label{eq:batt_lev}\\
& \qquad \qquad \quad  b_{j,t} = b^{\mathrm{EOD}}_{j}, \, t \in \{1\} \, \cup \, \mathcal{T}^{\mathrm{EOD}}, \, j \in \mathcal{J}, \label{eq:daily_cyc}     
 \end{align}
\end{subequations}
where $\pmb{\pi}_{0}$  (in \$/kWh) represents the TOU energy prices and $\tilde{\pi}_{0}^{\theta} := \pi_{0}^{\theta} / \Delta $ for $\theta \in \Theta_{0}$ with $\pi_{0}^{\theta} $ (in \$/kW) representing the demand rates across the TOU periods $\Theta_{0}$ at the charging station. For EV~$j$, $\mathbf{m}_{j} := \mathbf{1}~-~\mathbf{M}_{j}\mathbf{1}$, which takes the value 1 when the EV is not dispatched to any user's location, thereby is dispatched to the charging station and the value 0 if the EV is dispatched to a user's location. Here, $ \overline{p}_{0} $ denotes the charger capacity (in kW) installed at the charging station, $\eta^{\mathrm{ch}}_{j},  \eta^{\mathrm{dis}}_{j} \in (0,1]$ are the charging and discharging efficiencies respectively, $\overline{b}_{j}$ represents the EV's battery capacity, and $b_{j}^{\mathrm{EOD}}$ denotes the desired state of charge of EV battery by the end-of-day  intervals $\mathcal{T}^{\mathrm{EOD}}$ (e.g. midnight of each day of a month) for EV~$j$. 
Constraint~\eqref{ev:char} ensures that EVs can charge only at the charging station. This is because only when an EV (say EV~$j$) is placed at the charging station during $(t-1)$ and $t$ intervals accounting for transit, the value of $m_{j,t-1}m_{j,t} = 1$, otherwise $m_{j,t-1}m_{j,t} = 0$, forcing $e^{\mathrm{ch}}_{j,t} = 0$. Constraint~\eqref{eq:char_first_int} precludes EV charging during the first interval.
Constraint~\eqref{eq:batt_dyn} characterizes the dynamics of energy stored in the battery. Constraint~\eqref{eq:batt_lev} maintains the energy in the battery within the limit. Constraint~\eqref{eq:daily_cyc} initializes each EV's state of charge at
$b_{j}^{\mathrm{EOD}}$ and requires it to return to this desired level at
the end of each day.
\subsection{Mixed-integer linear reformulation}
One approach to solving the fleet operator's problem is to formulate a joint optimization that integrates the C\&I user's problem for each user and the EVs' operational problem, resulting into a mixed-integer non-linear  program (MINLP). Since MINLPs are known to be computationally expensive, we formulate it into a mixed integer linear program (MILP). The bilinear constraints~\eqref{alloc_mat:const_2}, \eqref{user:request}, and \eqref{ev:char} as well as the trilinear constraint \eqref{ev:discharge} can be linearized using standard McCormick reformulation~\cite{mccormick1976computability}. Specifically, we introduce an auxiliary variable $M_{j,t,i}^\prime = M_{j,t-1,i} \, M_{j,t,i}$, which can be linearized as
\begin{subequations}\label{mccor_linear}
\begin{align}
M_{j,t,i}^\prime & \geq  0,\\
M_{j,t,i}^\prime & \leq  M_{j,t-1,i},\\
M_{j,t,i}^\prime & \leq  M_{j,t,i},\\
M_{j,t,i}^\prime & \geq  M_{j,t-1,i} + M_{j,t,i} -1.
\end{align}
\end{subequations}

Substituting $M_{j,t,i}^\prime$ into \eqref{alloc_mat:const_2} and \eqref{user:request} leads to linear constraints. For the trilinear constraint \eqref{ev:discharge}, substituting $M_{j,t,i}^\prime$ results in a bilinear term, which can further be linearized in a similar manner as \eqref{mccor_linear}. Likewise, defining another auxiliary variable $m'_{j,t} = m_{j,t-1}m_{j,t}$ and applying the same reformulation enables \eqref{ev:char} to be expressed linearly.

Constraint \eqref{ev:transit_ener} is non-convex. To obtain a convex form of this constraint, we relax it by replacing the equality ``='' with an inequality ``$\geq$''. Lemma~\ref{lem:transit_energy} establishes that this relaxation is exact, i.e., the inequality is binding at the optimal solution since any excess transit energy beyond the required value would only increase charging and depreciation cost.

\begin{lemma}[Tightness of transit-energy relaxation]
\label{lem:transit_energy}
For a given EV dispatch matrix $\mathbf{M}_{j}$, the relaxation of \eqref{ev:transit_ener} into an inequality is binding at optimality.
\end{lemma}

\begin{proof}
The proof is provided in Appendix~\ref{app:proof_lemma1}.
\end{proof}

Even after reformulating the problem as an MILP which is NP-hard in general, and its computational complexity grows exponentially with the number of users and EVs. Consequently, conventional methods such as the branch-and-cut algorithm do not scale efficiently for large instances.

\section{Efficient Algorithm} \label{sec:HA}
We propose an algorithm to obtain a high-quality solution to the MILP defined in Section \ref{sec:FRM} in polynomial time. The algorithm follows a simple strategy that, at each iteration, dispatches the ``least used'' EV to the ``most-valued'' user's location when feasible and profitable without considering future iterations. 

Specifically, a single EV is dispatched to a single user for one or more intervals during each iteration based on the value these services provide. It is possible for a user previously served from an EV to be served again with a new EV for a new set of service intervals in the subsequent iterations. The algorithm terminates when no user yields a positive value for any set of service intervals.

The algorithm takes the energy consumption measurements of each user $\left(\pmb{\ell}_{i}\right)_{i \in \mathcal{I}}$ as the input, and outputs the dispatch matrices for each EV $\left(\mathbf{M}_{j}^{\mathrm{alg}}\right)_{j \in \mathcal{J}}$ and the variable that correspond to the service intervals for each user $\left(\mathbf{u}_{i}^{\mathrm{alg}}\right)_{i \in \mathcal{I}}$. For notational simplicity, we denote the collection of dispatch matrices as $\mathbf{M}^{\mathrm{alg}} := \left(\mathbf{M}_{j}^{\mathrm{alg}}\right)_{j \in \mathcal{J}}$ and the collection of the variables corresponding to service intervals as $\mathbf{u} ^{\mathrm{alg}}:= \left(\mathbf{u}_{i}\right)_{i \in \mathcal{I}}$. Using the dispatch matrices $\mathbf{M}^{\mathrm{alg}}$, we can obtain EVs charging schedules $\left(\mathbf{e}_{j}^{\mathrm{alg}}\right)_{j \in \mathcal{J}}$ and the corresponding state of charge $\left(\mathbf{b}_{j}^{\mathrm{alg}}\right)_{j \in \mathcal{J}}$ using \eqref{ev:charging_cost}. The pseudocode of the algorithm is shown in Algorithm~\ref{MatchEV2User}.

\begin{algorithm}[h]
\caption{Algorithm for EV dispatch and user service scheduling}
\label{MatchEV2User}
\SetAlgoLined
\KwInput{User load profiles $\pmb{\ell}_{i}$ for $i \in \mathcal{I}$}
\KwOutput{EV dispatch matrices $\mathbf{M}^{\mathrm{alg}}$, and the variables corresponding to service intervals $\mathbf{u}^{\mathrm{alg}}$}

Initialize $\mathbf{M}_{j}^{\mathrm{alg}} = \mathbf{0}$ for $j \in \mathcal{J}$ and $\mathbf{u}_{i}^{\mathrm{alg}} = \mathbf{0}$ for $i \in \mathcal{I}$\;
Compute initial priority values $V_{i,s}$ for all $i \in \mathcal{I}$ and $s \in \mathcal{S}$ using~\eqref{pri_val_users}\;

\While{any $V_{i,s} > 0$}{
    Identify the most valued user $i_\star$ and corresponding number of services $s_\star$ using~\eqref{most_val_user}\;
    Determine the service intervals $\mathcal{T}^{\mathrm{ser}}$ and corresponding transit intervals $\mathcal{T}^{\mathrm{tr}}$\;
    Determine the set of available EVs $\mathcal{J}_{\mathrm{a}}$ that are at the charging station during $\mathcal{T}^{\mathrm{ser}} \cup \mathcal{T}^{\mathrm{tr}}$\;

    \eIf{$\mathcal{J}_{\mathrm{a}} \neq \emptyset$}{
        Select the least used EV $j_\star$ from $\mathcal{J}_{\mathrm{a}}$ using~\eqref{least_utilized}\;

        \eIf{$s_\star V_{i_\star,s_\star} > \delta \mathrm{CC}$}{
            Update $M^{\mathrm{alg}}_{j_\star,t,i_\star} \gets 1$ for $t \in \mathcal{T}^{\mathrm{ser}} \cup \mathcal{T}^{\mathrm{tr}}$\;
            Update $u_{i_\star,t}^{\mathrm{alg}} \gets 1$ for $t \in \mathcal{T}^{\mathrm{ser}}$\;
            Recompute $V_{i_\star,s}$ for $s \in \mathcal{S}$ using~\eqref{pri_val_users}\;
        }{
            Set $V_{i_\star,s_\star} \gets 0$\;
        }
    }{
        Set $V_{i_\star,s} \gets 0$ for $s_\star \le s \le S$\;
        Recompute $V_{i_\star,s}$ for $1 \le s < s_\star$ using~\eqref{pri_val_users}\;
    }
}
\Return{$\mathbf{M}^{\mathrm{alg}}, \mathbf{u}^{\mathrm{alg}}$}\;
\end{algorithm}

Now, we explain the algorithm in detail. We first initialize the output variables $ \left(\mathbf{M}_{j}^{\mathrm{alg}} \right)_{j \in \mathcal{J}}$ and $ \left(\mathbf{u}_{i}^{\mathrm{alg}} \right)_{i \in \mathcal{I}}$ with zeros. This implies that initially, no EV is dispatched to any user's location, and no service intervals are designated for any user. Next, we discuss the three major steps of the algorithm, which are detailed as follows:

\textit{(a) Most-valued user selection:} 
In this step, we identify the user who offer the maximum demand charge reduction. To do so, we assign each user a priority-value, computed across 1 to $S$ services, based on their mean demand charge reduction per service. Formally, this value is given as 
\begin{align}\label{pri_val_users}
V_{i,s} = \left(\mathrm{DC}_{i} - \mathrm{DC}_{i}^{\mathrm{R}}(s)\right)/s \,\, \mathrm{for} \,\, s \in \mathcal{S} = \{1,\cdots, S\}.
\end{align}
Here, $\mathrm{DC}_{i} = \sum_{\theta \in \Theta} \tilde{\pi}_{i}^{\theta} \max_{t\in \mathcal{T}^{\theta}} \left(\ell_{i,t} - \overline{p}_{i}\Delta u^{\mathrm{alg}}_{i,t} \right)$ represents the demand charge with already assigned service intervals corresponding to $\mathbf{u}_{i}^{\mathrm{alg}}$ whereas  $\mathrm{DC}_{i}^{\mathrm{R}}(s)$ represents the reduced demand charge when $s$ additional services are provided to the user.  These additional service intervals corresponds to the variable $\mathbf{u}_{i}^{+}$. The reduced demand charge $\mathrm{DC}_{i}^{\mathrm{R}}$ can thus be obtained from the following optimization problem:
\begin{subequations}\label{prb:user_mod}
\begin{align}
\mathrm{DC}_{i}^{\mathrm{R}}(s) = & \min_{\mathbf{u}^{+}_{i}} \,\, \sum_{\theta \in \Theta} \tilde{\pi}_{i}^{\theta} \max_{t\in \mathcal{T}^{\theta}} \left(\ell_{i,t} - \overline{p}_{i}\Delta \left(u^{\mathrm{alg}}_{i,t} + u^{+}_{i,t} \right)\right)\\
& \textrm{s.t.} \quad \mathbf{u}^{+}_{i} \in \{0,1\}^{T}, \label{bin_var}\\
& \qquad u^{+}_{i,t} = 0, \,  t \in \mathcal{O}_{i},\label{ev_occupancy}\\
& \qquad \mathbf{1}^\intercal \mathbf{u}_{i}^{+} \leq s,\label{add_ser}
\end{align}
\end{subequations}
where $\mathcal{O}_{i} = \{t \in \mathcal{T}: u_{i,t}^{\mathrm{alg}} = 1 \mbox{ or } \sum_{j \in \mathcal{J}, i \in \mathcal{I}}M_{j,t,i}^{\mathrm{alg}} = J \mbox{ or } \ell_{i,t} \leq \overline{p}_{i}\Delta  \}$.
Constraint~\eqref{bin_var} indicates that $\mathbf{u}_{i}^{+}$ is a binary variable corresponding to additional service intervals. 
Constraint~\eqref{ev_occupancy} prohibits service during any interval that satisfies one or more of the following conditions: (i) the user has already been served  $\left( u_{i,t}^{\mathrm{alg}} = 1 \right) $, (ii) all the EVs are fully occupied $\left( \sum_{j \in \mathcal{J}, i \in \mathcal{I}} M_{j,t,i}^{\mathrm{alg}} = J \right)$, or (iii) the user’s energy consumption is lower than the energy that would be delivered by an EV operating at the charger’s maximum power capacity $\left( \ell_{i,t} \leq \overline{p}_{i}\Delta  \right)$, thereby avoiding any return of energy to the grid. Constraint~\eqref{add_ser} limits the total number of services. The constraint matrix of \eqref{prb:user_mod} is totally unimodular with an integral right-hand side, making the feasible region an integral polyhedron. Thus, the LP relaxation yields an integer solution without explicitly enforcing Constraint~\eqref{bin_var}, which can be replaced by convex constraint $\mathbf{0} \leq \mathbf{u}_{i}^{+} \leq \mathbf{1}$.

Based on the priority-values $\mathbf{V}$, we select the most-valued user and the number of services as \begin{align}\label{most_val_user}
\left( i_{\star}, s_{\star} \right) = \mbox{arg} \max_{i \in \mathcal {I},s\in \mathcal {S}} V_{i,s}
\end{align} such that $V_{i_\star,s_\star} >0$. The corresponding service intervals can then be obtained as $\mathcal{T}^{\mathrm{ser}} = \left\{ t \in \mathcal{T} : u_{i_\star,t}^{+\star}
= 1 \right\}$ where $u_{i_\star,t}^{+\star}$ is an optimal solution of \eqref{prb:user_mod} when solving for user $i_\star$ for $s_\star$ services.  

It is possible that there might be more than one combination of $(i_\star,s_\star)$  which corresponds to the maximum priority-value. In those cases, we break the tie as follows: (i) the user with the maximum number of services will be given the priority. This is because more number of services correspond to more demand charge reduction for the same mean demand charge reduction. (ii) If the ties still persist, the user closest to the charging station will be given the priority. This ensures the energy spend for the transit is the lowest.

\textit{(b) Least used available EV:} 
In this step, we identify the EV that will deliver services to the most-valued user $i_\star$ during the service intervals $\mathcal{T}^{\mathrm{ser}}$. To do so, we first determine the corresponding transit intervals $\mathcal{T}^{\mathrm{tr}}$, defined as the intervals immediately preceding the service intervals, i.e., $\mathcal{T}^{\mathrm{tr}} = \{t-1: t\in \mathcal{T}^{\mathrm{ser}} \}$. We then obtain the set of available EVs $\mathcal{J}_{\mathrm{a}}$ that are located at the charging station during both the service and transit intervals, based on their spatial coordinates $\left(\mathbf{x}_{j}, \mathbf{y}_{j} \right)$ derived from \eqref{ev:x} and \eqref{ev:y}. Mathematically, it can be given as
\begin{align}\label{avail_evs}
\mathcal{J}_{\mathrm{a}} = \{ j \in \mathcal{J}: \left(x_{j,t},y_{j,t}\right) = \left(x_{0}, y_{0} \right) \textrm{ for } t \in \mathcal{T}^{\mathrm{ser}} \cup \mathcal{T}^{\mathrm{tr}}  \}.
\end{align}

We then select the least-used EV $j_\star$ among the available EVs $\mathcal{J}_{\mathrm{a}}$, defined as the vehicle with the lowest total energy expenditure across discharging and transit. Mathematically, it can be given as
\begin{align}\label{least_utilized}
j_\star = \mathrm{arg} \min_{j \in \mathcal{J}_{\mathrm{a}}}\, 
(\mathbf{e}_{j}^{\mathrm{dis}} + \mathbf{e}_{j}^{\mathrm{tr}})^{\intercal}\mathbf{1}.
\end{align}
This selection promotes a more balanced distribution of energy usage across the fleet, thereby increasing the likelihood of minimizing the demand charge at the charging station. If there are multiple EVs satisfy the condition, we select one randomly to break the tie.

It is possible that the set of available EVs is empty. It implies that none of the EVs can single-handedly serve the most-valued user for the required service intervals. We then update the corresponding priority-value with zero i.e. $V_{i_\star, s_\star} = 0$. More importantly, this also implies that if $s_\star$ services are infeasible for the user, services more than $s_{\star}$ are also infeasible. Thus, we update the corresponding priority values with zeros\footnote{Although infeasible service intervals for the most-valued user are assigned zero priority in the current iteration, this infeasibility only rules out serving these intervals with a single EV. Subsets of these intervals may become feasible in later iterations when served by different EVs.}, i.e.,  $V_{ i_{\star},s} = 0$ for $s_\star \leq s \leq S$. We then update the priority values $V_{i_\star,s}$ for $1 \leq s < s_\star$ using \eqref{pri_val_users}. The algorithm subsequently proceeds to select the next most valued user using \eqref{most_val_user} and to identify the corresponding least used available EV using \eqref{avail_evs} and \eqref{least_utilized}.
 
\textit{(c) Marginal value-based dispatch: } In this step, we dispatch the least used available EV~$j_\star$ at the most-valued user~$i_\star$'s location during the service intervals $\mathcal{T}^{\mathrm{ser}}$ and the corresponding transit intervals $\mathcal{T}^{\mathrm{tr}}$ if the marginal demand charge reduction $s_\star V_{i_\star,s_\star}$ exceeds the corresponding increase in charging cost $\delta \mathrm{CC}$. The increase in the charging cost is computed as \[ \delta \mathrm{CC} = \mathrm{CC}\left(\mathbf{M}^{\mathrm{alg}+}\right) - \mathrm{CC}\left(\mathbf{M}^{\mathrm{alg}}\right)\] where $\mathrm{CC}(\cdot)$ can be obtained using \eqref{ev:charging_cost}. Here, $\mathbf{M}^{\mathrm{alg}+}$ represents the set of dispatch matrices after tentatively dispatching EV $j_\star$ at the user $i_\star$'s location. Specifically, all other EVs retain their dispatch matrices, i.e., $\mathbf{M}_{j}^{\mathrm{alg}+} = \mathbf{M}_{j}^{\mathrm{alg}}$ for $j \neq j_\star$. For the selected EV~$j_\star$, $\mathbf{M}_{j_\star}^{\mathrm{alg}+} = \mathbf{M}_{j_\star}^{\mathrm{alg}}$ except during the intervals $\mathcal{T}^{\mathrm{ser}}$ and $\mathcal{T}^{\mathrm{tr}}$, where $M_{j_{\star}, t, i_{\star}}^{\mathrm{alg}+} = 1$. 

In other words, if $s_\star V_{i_\star,s_\star} > \delta\mathrm{CC}$, we update the dispatch matrix for EV $j_\star$ with $M_{j_\star,t,i_\star}^{\mathrm{alg}} = 1$ for $t \in \mathcal{T}^{\mathrm{ser}} \cup \mathcal{T}^{\mathrm{tr}}$ and the variable corresponding to the most-valued user $i_\star$'s service intervals $u_{i_\star,t}^{\mathrm{alg}} = 1$ for $t \in \mathcal{T}^{\mathrm{ser}}$. We then update the priority-values for user $i_\star$ using \eqref{pri_val_users} with the updated dispatch matrices $\mathbf{M}^{\mathrm{alg}}$ and the updated service intervals represented by $\mathbf{u}_{i_\star}^{\mathrm{alg}}$. Conversely, if $s_\star V_{i_\star,s_\star} \leq \delta \mathrm{CC}$, the corresponding priority value is set to zero, $V_{i_\star,s_\star} = 0$. The algorithm subsequently proceeds to select the next most valued user using \eqref{most_val_user} and to identify the corresponding least used available EV using \eqref{avail_evs} and \eqref{least_utilized}.

\textit{Stopping criterion:} The algorithm repeats Steps (a) to (c) iteratively, updating the priority values after each iteration. The process terminates when all priority values become zero, that is, 
$V_{i,s} = 0$ for all $i \in \mathcal{I}$ and $s \in \mathcal{S}$ indicating that no further user and EV pair can yield a positive marginal demand charge reduction. The priority values in Step~(a) are updated only for users affected by the most recent dispatch, such as the most valued user $i_\star$ just served or those with modified feasible intervals. Updating only these users significantly reduces computation time without affecting the algorithm's performance.

\textit{Computational complexity:} Rejected or infeasible candidate pairs $(i,s)$ are permanently removed once $V_{i,s}$ is set to zero, yielding at most $\mathcal{O}(IS)$ iterations. Successful dispatches may revisit a user after priority recomputation, but each dispatch assigns at least one previously unserved interval, limiting their number to $\mathcal{O}(IT)$. Thus, the outer loop executes at most $\mathcal{O}(IS+IT)$ iterations. Each iteration requires $\mathcal{O}(IS)$ time to identify the highest-priority candidate pair, $\mathcal{O}(JT)$ time to verify EV availability across the scheduling horizon, and $\mathcal{O}(J)$ time to select the least-utilized feasible EV. In addition, priority updates require solving up to $S$ instances of \eqref{prb:user_mod}, each solvable in $\mathrm{poly}(T)$ time due to total uni-modularity of the constraint matrix. Therefore, the per-iteration complexity is $\mathcal{O}(IS+JT+J+S\mathrm{poly}(T))$. Therefore, the worst-case complexity is
$\mathcal{O} \left((IS+IT)\left(IS+JT+J+S\mathrm{poly}(T)\right)\right).$

\textit{Theoretical performance:} The proposed algorithm is a myopic heuristic and no universal approximation guarantee is claimed. Decisions are made sequentially and are not revisited, so locally beneficial assignments may preclude more valuable future opportunities. Moreover, deriving a formal approximation guarantee is challenging because the problem does not admit a simple separable and additive structure. The demand-charge objective involves a maximum operator, and assignment values are coupled through EV availability, SOC dynamics, charging decisions, and transit requirements. Therefore, although the proposed heuristic produces high-quality feasible solutions empirically, no general theoretical approximation bound is established.

\section{Case Study}\label{sec:CS}

\subsection{Study inputs and computational setup}

\subsubsection{Data description} We use smart meter data\footnote{The smart-meter data used in this study are subject to a non-disclosure agreement and cannot be publicly released. Preprocessing scripts are available in the authors' public repository \cite{reddy2026github}, and the DOE/NREL ComStock dataset \cite{comstock2024} may be used as a publicly available alternative.}  from 730 business buildings across 24 zip codes in San Francisco, California. Electricity consumption is recorded at 15-minute intervals over 27 months, from May 2014 to July 2016. The dataset covers nearly the entire city, excluding Treasure Island, Presidio, and Mission Bay. All buildings are served electricity by PG\&E Company. 

\subsubsection{User tariff structure} We classify building users under PG\&E’s tariff structure. PG\&E evaluates tariff eligibility on a monthly basis. In this study, however, we focus on users with an average monthly peak demand above 75 kW, for whom demand rates apply. A total of 136 users satisfy this criterion. These users are then assigned to either Schedule B-10 or B-19. Users with average peak demand exceeding 499 kW are placed in Schedule B-19, and the rest in Schedule B-10. This results in 78 B-10 users and 58 B-19 users \cite{pge2025}.

Schedule B-10 applies a flat demand rate of \$22.16/kW. There is no TOU differentiation. Schedule B-19 applies a base demand rate of \$39.22/kW. Additional increments depend on both TOU period and season. In summer (June--September), the incremental rates are \$54.17/kW during peak hours (4--9 pm) and \$11.75/kW during partial-peak hours (2--4 pm and 9--11 pm). In winter (October--May), the incremental rate is \$3.20/kW during peak hours (4--9 pm).

In addition to tariff classification, we also account for the spatial distribution of the users. The building locations of these users are available at the zip code level, with synthetic locations randomly assigned within the boundary of each zip code. Consequently, the estimated travel distances and transit energy costs should be interpreted as approximations. When higher-resolution building-location and road-network data are available, they can be directly incorporated into the analysis. The charging station is located at the geometric center of the buildings, computed as the mean of their coordinates, as illustrated in Fig.~\ref{location}. The geographical coordinates (latitude and longitude) of the users and the charging station are projected onto a local Cartesian plane by converting them into \textit{x-y} coordinates using the Haversine formula~\cite{haversine_formula_wikipedia}, enabling subsequent distance calculations.
\begin{figure}[htbp]
\centering
\includegraphics[width=0.5\textwidth]{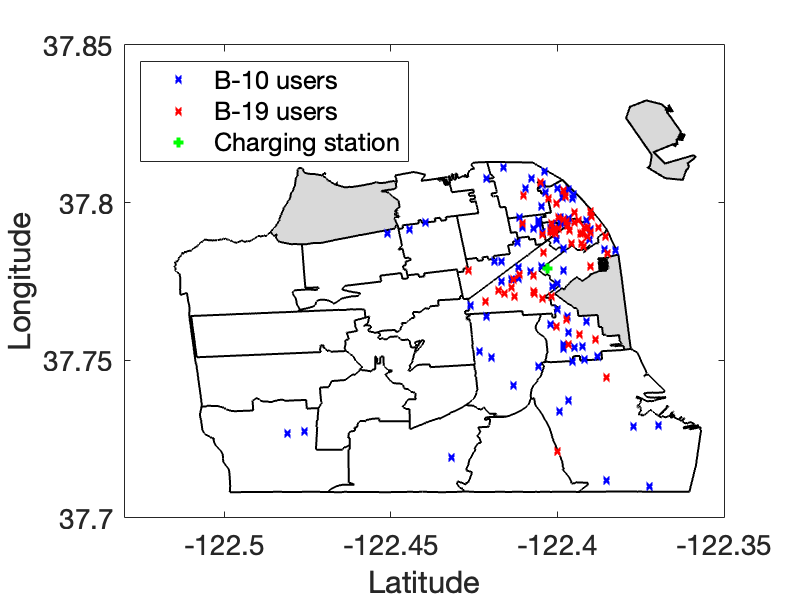}
\caption{Map with user locations and charging station.}
\label{location}
\end{figure} 

\subsubsection{EV parameters and charging tariffs} We consider Nissan Leaf as a representative EV example for this business model. For each EV $j \in \mathcal{J}$, the parameters are listed in Table~\ref{tab:ev_parameters}.  As all EVs are considered to charge exclusively at the charging station, we adopt Schedule BEV for its tariff structure \cite{pge2025}. This tariff does not incorporate either TOU or seasonal differentiation in its demand rates. The demand rate, referred to as the subscription charge, is \$1.24/kW. The energy rates, however, are differentiated by TOU period: \$0.38/kWh for peak hours (4-9 pm), \$0.19/kWh for off-peak hours (9 pm-9 am and 2-4 pm), and \$0.16/kWh for super off-peak hours (9 am-2 pm).

\begin{table}[htbp]
\centering
\caption{Nissan Leaf (S) parameters \cite{9203459}, \cite{nissan_leaf_buildprice, raustad2014electric}}
\label{tab:ev_parameters}
\begin{tabular}{l c c}
\hline
\textbf{Parameter} & \textbf{Symbol} & \textbf{Value} \\
\hline
Maximum SOC limit of EV $j$     & $\overline{b}_{j}$ & 30 kWh \\
Minimum SOC limit of EV $j$     & $\underline{b}_{j}$ & 10 kWh \\
End-of-day SOC of EV $j$ & $b_{j}^{\mathrm{EOD}}$ & 20 kWh\tablefootnote{We set the end-of-day SOC below the maximum limit to reduce demand charge impacts at the charging station. Because services are unlikely during early morning hours, this choice provides additional charging flexibility while still allowing the EV to reach full SOC before service provision.} \\
Charging efficiency of EV $j$ & $\eta_{j}^{\mathrm{ch}}$ & 87\% \\
Discharging efficiency of EV $j$ & $\eta_{j}^{\mathrm{dis}}$ & 87\% \\
Range efficiency of EV $j$ & $\delta^{\mathrm{tr}}_{j}$ & 3.5 miles/kWh\\
\multicolumn{3}{l}{Depreciation rate of EV $j$}            \\
\,\,\, Usage-based & $\gamma_{j}$ & \$ 0.74/kWh\\
\,\,\, Age-based & & \$189/month\\
\hline
\end{tabular}
\end{table}

\textit{Parameter sources and derivations:} The battery capacity and vehicle range are obtained from the top-right table and bottom-left table specification, respectively, in~[\cite{nissan_leaf_buildprice}, pg.~8]. The minimum and maximum states of charge, $\underline{b}_{j}$ and $\overline{b}_{j}$, are assumed within the battery capacity limits. The range efficiency is derived as $\delta^{\mathrm{tr}}_{j}=150/40=3.75$~miles/kWh. Reference~\cite{9203459} reports round-trip efficiencies of $79.1\%$--$87.8\%$. Assuming equal charging and discharging efficiencies and using the lower end gives approximately $87\%$ for each.

The age-based depreciation rate is derived using an initial vehicle cost of \$31{,}415~[\cite{raustad2014electric}, pg.~3], annual depreciation values of \$6{,}515, \$2{,}881, \$2{,}535, \$2{,}248, and \$2{,}017 for the first five years~[\cite{raustad2014electric}, pg.~27], and an assumed residual value of $1.5\%$ after 20 years, based on~[\cite{raustad2014electric}, pg.~8]. Intermediate vehicle values are estimated using shape-preserving piecewise cubic interpolation, and annual depreciation is averaged over the first ten years resulting \$2{,}263 or \$189 monthly. The usage-based depreciation rate is calculated by subtracting the resulting average annual depreciation from the annual ownership cost of \$4{,}683 [\cite{raustad2014electric}, pg.~1], dividing by the annual mileage of 12{,}330 miles~[\cite{raustad2014electric}, pg.~1], and multiplying by $\delta^{\mathrm{tr}}_{j}$, resulting \$0.74/kWh.

\subsubsection{Estimation of labor cost} Labor costs are estimated by analyzing the temporal distribution of user sub-peaks events, where sub-peaks denote the $R$ highest load intervals in a month for each user. We choose $R = 10$, to reflect the expectation that peak-shaving events for an individual user are rare, with fewer than ten service intervals anticipated per month. For this choice of $R$, we aggregate peak occurrences across all users and map them onto a time-of-day and day-of-week representation. Averaging this profile across weeks yields a heat map as shown in Figure~\ref{ev_earn_mult_evs} that reflects the expected number of users experiencing peak events at each time period. Since the majority of services are scheduled during weekday business hours (8 am to 6 pm), each EV is assumed to be operated by a dedicated driver working ten hours per weekday. Accordingly, a monthly labor cost of \$4,180 per driver (based on \$19 per hour) is incorporated into the cost structure, consistent with average private driver wages in San Francisco~\cite{ziprecruiter2025driverSF}. This assumption provides a simplified and conservative baseline for operational cost estimation. In practical deployments, alternative operational strategies such as shared-driver scheduling, centralized fleet coordination, and optimized routing may reduce the effective labor cost per EV.

\begin{figure}[h]
\centering
\includegraphics[width=0.48\textwidth]{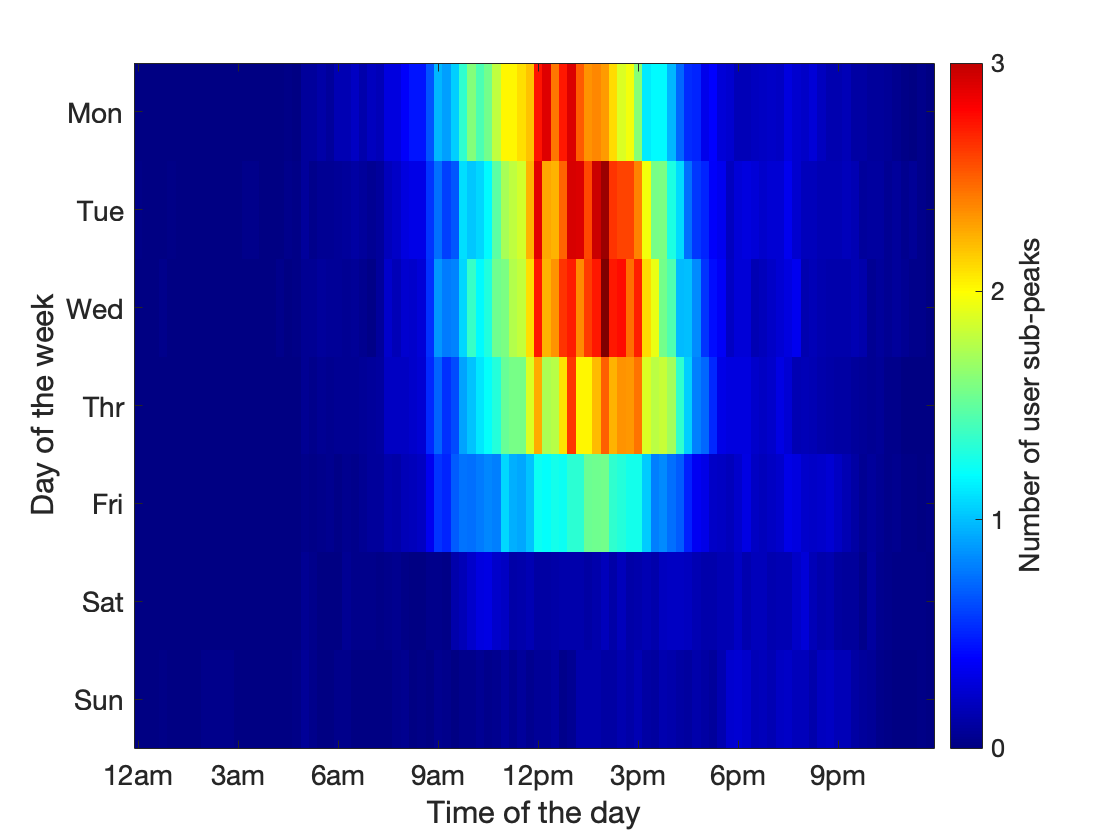}
\caption{Temporal distribution of user sub-peaks with $R=10$.}
\label{ev_earn_mult_evs}
\end{figure} 

\subsubsection{Computational setup} We use a high-performance computing system equipped with an AMD 7985WX processor (64 cores, 128 threads, 3.2–5.1 GHz) and 512 GB of RAM. Evaluations are performed in MATLAB R2023a using CVX 2.2 and Gurobi 13.0.0. All settings remain at their defaults, including relative and absolute MIP gap tolerances of $10^{-4}$ and $10^{-10}$, respectively, except for a five-hour time limit.

\subsection{Single EV scenario}

We begin with the case where a single EV provides services to all users equipped with AC Level-2 chargers installed at their locations. Each service is delivered at the charger’s rated capacity of $\overline{p}_{i} = 15~\mathrm{kW}$ per interval~\cite{steward2017critical}.  We solve the MILP, which requires an average of 96 minutes for each monthly instance, whereas the proposed algorithm with $S=10$ obtained a solution in under~5 minutes. The mean sub-optimality gap between the optimal value and the value obtained from proposed algorithm is 4.8\%.
Since the algorithm yields a feasible solution that may be suboptimal, with a small sub-optimality gap, all results presented in the subsequent sections are obtained using the algorithm.

To evaluate the profitability of the business model, we define the net savings as the total demand charge reduction across all users after accounting for EV operational cost, age-based depreciation, charger costs, and labor costs. The amortized monthly cost associated with each AC Level-2 charger is \$19.17~\cite{smith2015costs}. Figure~\ref{profitMonths} presents the net savings, which demonstrate the profitability of the business model and the close performance of the proposed algorithm relative to the optimal solution. The results also reveal a clear seasonal pattern, with average monthly summer savings of \$34,111 (highlighted in yellow), nearly double the winter monthly average of \$18,025 (highlighted in cyan).

 \begin{figure}[h]
\centering
\includegraphics[width=0.5\textwidth]{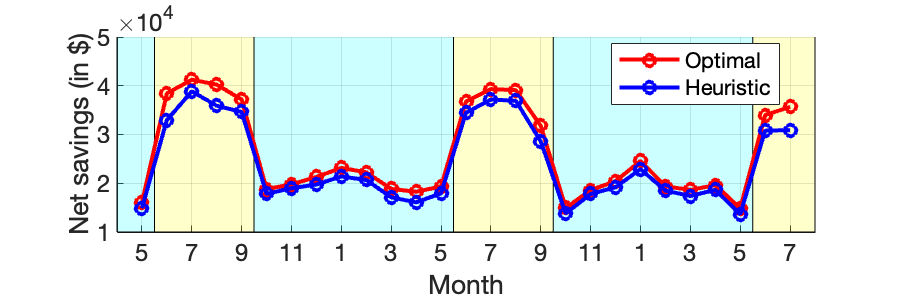}
\caption{Net savings with a single EV.}
\label{profitMonths}
\end{figure} 

To further understand the seasonal variation in net savings, we first examine the number of services provided to each user type, as shown in Figure~\ref{fig:dist_prov_ints}. B-10 users receive a median of 93 services in winter and 79 in summer, while B-19 users receive 80 and 99, respectively. The total number of services across both seasons is comparable (around 175), suggesting that the seasonal difference in net savings arises from variations in the magnitude and timing of demand reductions rather than service frequency. 
\begin{figure}
    \centering
    \includegraphics[width=0.45\textwidth]{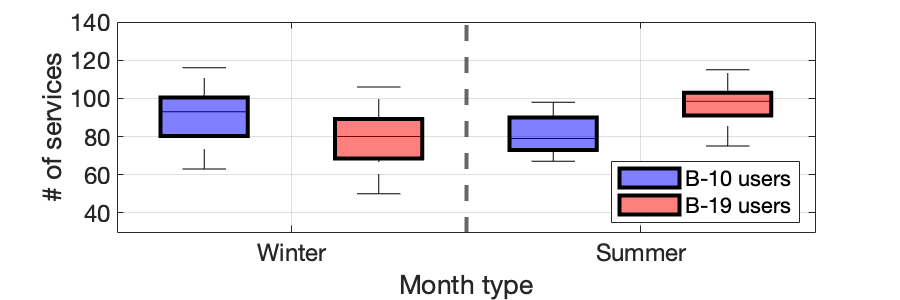}
    \caption{Service frequency by user type.}
    \label{fig:dist_prov_ints}
\end{figure}
As shown in Figure~\ref{fig:dis_dem_red_tou}(a), the mean demand reduction per service for B-10 users is 2.57~kW in winter and 2.25~kW in summer, and for B-19 users, 6.54~kW in winter and 4.62~kW in summer. The lower summer demand reduction relative to its winter counterpart for B-19 users may seem counterintuitive given the higher number of services. However, this outcome reflects the tariff structure, where reductions during peak and partial peak hours contribute more substantially to overall savings (see Figure~\ref{fig:dis_dem_red_tou}(b)). 
\begin{figure}[htbp]
    \centering
    \begin{subfigure}{0.45\textwidth}
        \centering
        \includegraphics[width=\linewidth]{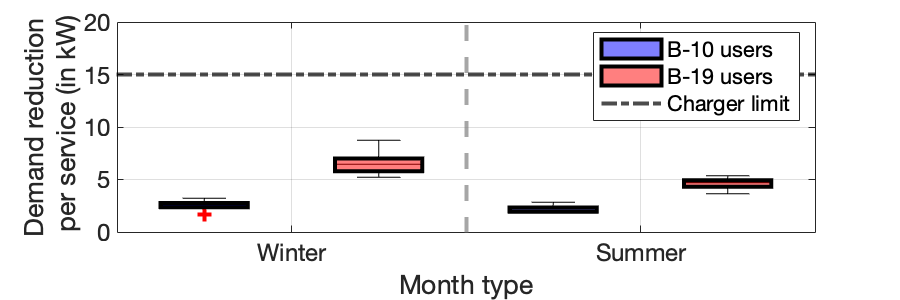}
        \caption{Peak demand reduction per service.}
        \label{fig:dem_red}
    \end{subfigure}
    \begin{subfigure}{0.45\textwidth}
        \centering
        \includegraphics[width=\linewidth]{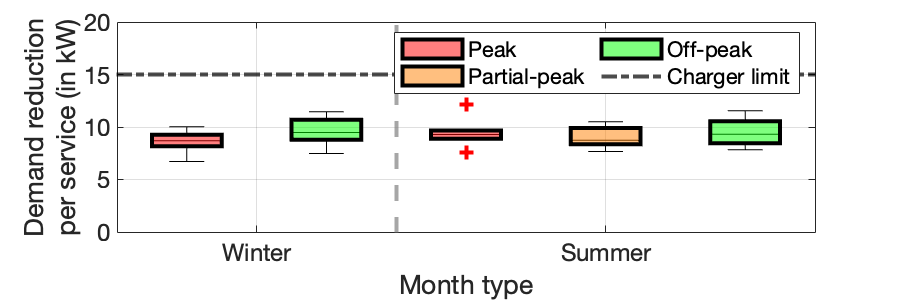}
        \caption{Peak demand reduction per service for B-19 user by time periods.}
        \label{fig:dem_red_tou}
    \end{subfigure}
    \caption{Peak demand reduction per service by season and user type.}
    \label{fig:dis_dem_red_tou}
\end{figure}
This effect is evident in Figure~\ref{fig:dem_charge_red}, where B-10 users achieve mean monthly demand charge reductions of \$66.05 in winter and \$51.33 in summer, while B-19 users achieve \$355.03 and \$652.67, respectively. After accounting for the monthly amortized charger cost of \$19.17, B-10 users can allocate up to \$39.31 per service in winter and \$31.75 in summer, whereas B-19 users can allocate up to \$243.50 and \$371.14 per service, respectively.
\begin{figure}
    \centering
    \includegraphics[width=0.45\textwidth]{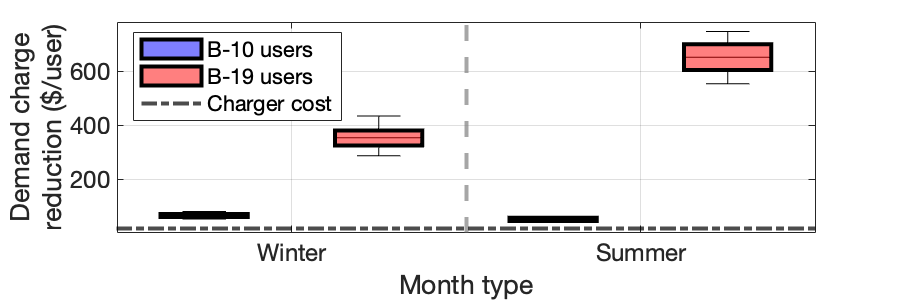}
    \caption{Demand charge reduction by season and user type.}
    \label{fig:dem_charge_red}
\end{figure}

On the cost side, the average monthly EV operating cost is \$759.26, including \$18.24 for demand charges, \$180.68 for energy charges, and \$560.34 for usage-based depreciation, with no significant seasonal variation. Adding the amortized charger cost of \$19.17, the age-based depreciation of \$189, and the labor cost of \$4180, the monthly cost to the fleet operator is \$5147.43. With an average of 175 services per month, the break-even price corresponds to \$29.41 per service.

\subsubsection*{Impact of infrastructure configuration} Apart from the current setup (All-AC), we consider two additional setups: (a) All-DC, where every user is equipped with a 30 kW DC fast charger with an amortized monthly cost of approximately \$133.33 \cite{steward2017critical,smith2015costs}, and (b) Tiered, where B-10 users have AC Level-2 chargers and B-19 users have DC fast chargers. In both setups, the charging station is equipped with a DC fast charger. The resulting net savings are presented in Figure~\ref{savings_setups}. Among the three configurations, the Tiered setup yields the highest savings across both seasons. Between All-DC and All-AC, the latter provides higher savings in winter, whereas the former is more beneficial during summer.

 \begin{figure}[h]
\centering
\includegraphics[width=0.5\textwidth]{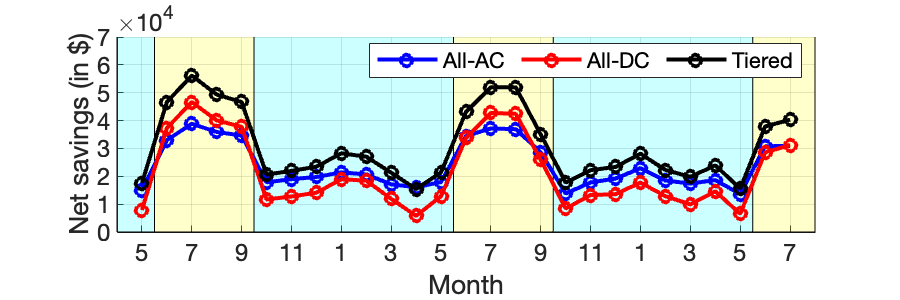}
\caption{Net savings under different setups.}
\label{savings_setups}
\end{figure} 

To further understand net savings across setups, we examine the demand charge reductions shown in Figures~\ref{dcr_season}(a) and~\ref{dcr_season}(b). In the All-DC setup, B-10 users achieve reductions comparable to All-AC, which are insufficient to offset the higher charger cost, whereas B-19 users gain significantly, with reductions about 1.5 times higher than in All-AC during both seasons. The Tiered setup combines the strengths of both, yielding the highest overall savings. On the cost side, the average monthly EV operating expense (see Figure~\ref{avg_month_ev_cost}) increases by about \$500 in the All-DC and Tiered setups compared to All-AC. But this can be offset by the savings from just two to three B-19 users. As All-DC remains impractical for B-10 users, we focus on the Tiered setup. During winter, the EV provides a median of 88 services to B-10 users and 77 to B-19 users, and in summer, 74 and 95, respectively. To break even, the fleet operator must recover an average monthly cost of \$5,523.82 across seasons. Since each service provided to a B-19 user involves twice the energy delivery of a service provided to a B-10 user, we adopt a differentiated per-service pricing structure. Under this scheme, the break-even price is \$22.82 per service for B-10 users and \$45.64 per service for B-19 users during winter, and \$20.92 and \$41.84 per service, respectively, during summer.

\begin{figure}[htbp]
    \centering
    \begin{subfigure}{0.48\textwidth}
        \centering
        \includegraphics[width=\linewidth]{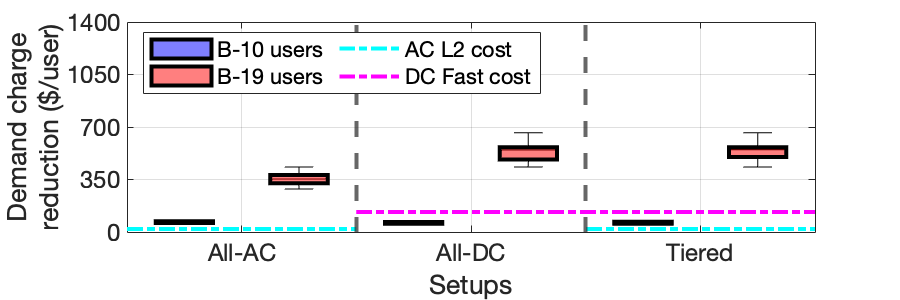}
        \caption{Demand charge reduction during winter months.}
    \end{subfigure}
    \hfill
    \begin{subfigure}{0.48\textwidth}
        \centering
        \includegraphics[width=\linewidth]{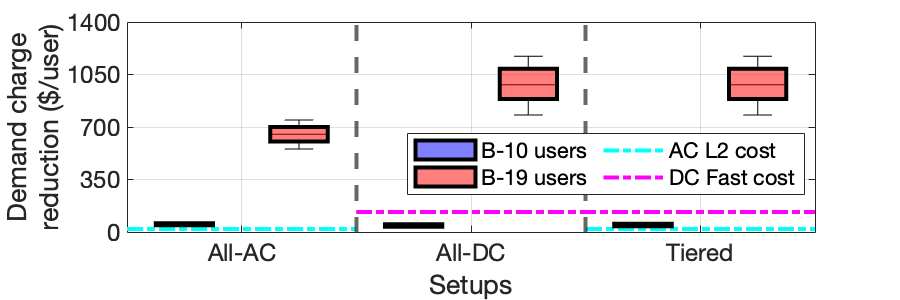}
        \caption{Demand charge reduction during summer months.}
    \end{subfigure}
    \caption{Demand charge reduction under different infrastructure configuration by season and user type.}
    \label{dcr_season}
\end{figure}

 \begin{figure}[h]
\centering
\includegraphics[width=0.45\textwidth]{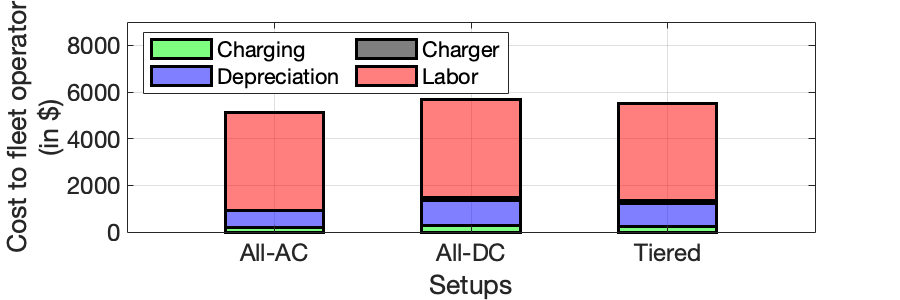}
\caption{Breakdown of the monthly cost to fleet operator.}
\label{avg_month_ev_cost}
\end{figure} 

\subsubsection*{Justification of the single-interval transit assumption}

Transit between consecutive service locations is assumed to be completed within one 15-min interval. This assumption is reasonable for this case study because the C\&I users are geographically concentrated in a dense urban area. Figure~\ref{fig:dist_transit} shows the distribution of inter-service transit distances obtained from the optimized dispatch schedule for a representative month under the tiered charging configuration with a single EV. The resulting distances are generally short, with a mean distance of approximately 1.34~miles. The single-EV case is a restrictive setting because the same EV must sequentially travel between assigned service locations. In the multiple-EV case, the platform has greater spatial flexibility to assign requests to nearby EVs, so the average transit distance per service is expected to be smaller. These observations support the use of the single-interval transit assumption for the considered service region.

\begin{figure}[htbp]
    \centering
    \includegraphics[width=0.5\linewidth]{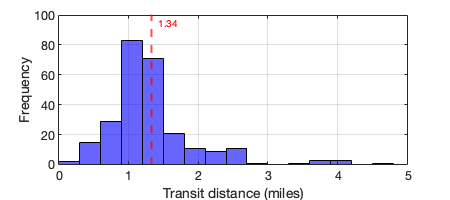}
    \caption{Distribution of transit distances with single EV under tiered setup.}
    \label{fig:dist_transit}
\end{figure}

\subsubsection*{Sensitivity to labor cost}

The sensitivity of net savings to labor cost for the tiered configuration is shown in Figure~\ref{sen_ana_labor_cost}. Labor rates are varied from \$0/hr, representing a fully autonomous vehicle operation with no driver labor cost, to \$300/hr. As expected, net savings decrease with increasing labor rate. The critical labor cost is approximately \$120/hr in winter and \$210/hr in summer, beyond which net savings become negative. The dotted red line denotes a baseline labor cost of \$19/hr, well below the critical values in both seasons.
 \begin{figure}[htbp]
\centering
\includegraphics[width=0.5\textwidth]{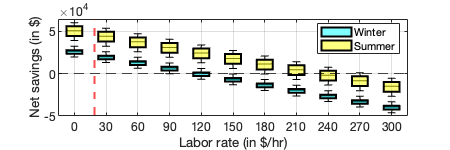}
\caption{Net savings sensitivity to labor cost.}
\label{sen_ana_labor_cost}
\end{figure}

\subsubsection*{Sensitivity to depreciation cost}

We also perform a sensitivity analysis of the net savings to the depreciation cost by varying both the usage-based and age-based depreciation rates. Figure~\ref{fig:sen_ana_dep_cost} illustrates the resulting net savings during winter and summer seasons under a tiered setup. The usage-based depreciation rate is varied from \$0.5/kWh to \$1/kWh, while the age-based depreciation rate is varied from \$150/month to \$250/month.

The net savings decrease linearly as either the usage-based or age-based depreciation rate increases. The impact of the usage-based depreciation rate is more pronounced during summer than winter because summer demand-charge rates are substantially higher, increasing the economic value of battery utilization for peak reduction. Consequently, EVs are dispatched more frequently and experience greater battery throughput during summer. As a result, each incremental increase in the usage-based depreciation rate leads to a larger reduction in summer net savings than in winter net savings. In contrast, the age-based depreciation cost is independent of battery utilization and therefore affects winter and summer net savings equally. Nevertheless, the achievable net savings remain relatively robust across the considered parameter ranges. In particular, break-even conditions occur only at substantially higher depreciation rates, indicating that the economic feasibility of the proposed framework is relatively insensitive to realistic variations in battery depreciation costs.
\begin{figure*}[htbp]
\centering

\begin{subfigure}[b]{0.48\textwidth}
    \centering
    \includegraphics[width=\linewidth]{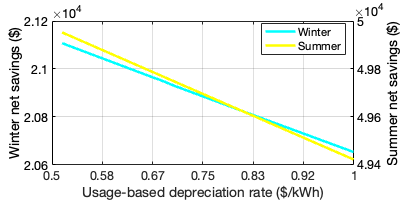}
    \caption{Usage-based}
    \label{fig:winter_dep}
\end{subfigure}
\hfill
\begin{subfigure}[b]{0.48\textwidth}
    \centering
    \includegraphics[width=\linewidth]{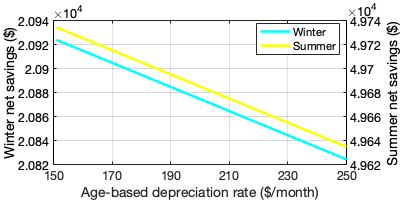}
    \caption{Age-based}
    \label{fig:summer_dep}
\end{subfigure}

\caption{Sensitivity plots for depreciation cost with a single EV under tiered setup.}
\label{fig:sen_ana_dep_cost}
\end{figure*}

\subsection{Multiple EV scenario}

Next, we examine the case where multiple EVs provide the services. The MILP could not reach an optimal solution even for two EVs within a five-hour limit, so all results in this section are obtained using the proposed algorithm with $S = 10$.  Although computation time increases with fleet size, it remains manageable for practical use. For instance, running the full optimization for twenty EVs takes about one hour.

The net savings for winter and summer months across the three setups are shown in Figures~\ref{net_savings_mult_evs}(a) and \ref{net_savings_mult_evs}(b), respectively. The Tiered setup consistently yields higher net savings than both All-AC and All-DC in each season. Considering the average net savings across the months, the best-performing fleet sizes during winter is two EVs for All-AC, and three EVs for both All-DC and Tiered setups. In summer, the best-performing fleet sizes increase to four for All-AC and six for both All-DC and Tiered, reflecting greater service opportunities and higher demand charge reduction potential.

\begin{figure}[htbp]
    \centering
    \begin{subfigure}{0.48\textwidth}
        \centering
        \includegraphics[width=\linewidth]{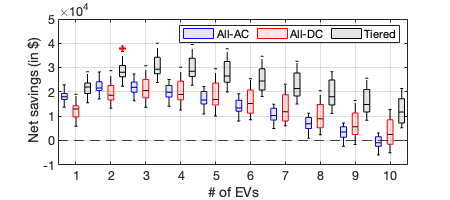}
        \caption{Net savings during winter months.}
        \label{net_savings_mult_evs_winter}
    \end{subfigure}
    \hfill
    \begin{subfigure}{0.48\textwidth}
        \centering
        \includegraphics[width=\linewidth]{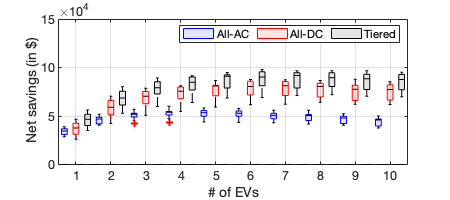}
        \caption{Net savings during summer months.}
        \label{net_savings_mult_evs_summer}
    \end{subfigure}
    \caption{Net savings with number of EVs under different infrastructure configurations by season.}
    \label{net_savings_mult_evs}
\end{figure}

The underlying demand charge reductions for B-10 and B-19 users during summer and winter are shown in Figure~\ref{dcr_mult_ev_all_setups}. As the number of EVs increases, the total demand charge reduction initially rises due to improved service coverage but gradually plateaus beyond a certain fleet size, regardless of the season or user type. This behavior arises because each service interval provides a fixed amount of discharge energy determined by the bidirectional charger rating. Consequently, the first EVs are typically deployed to mitigate the highest-demand intervals that contribute most to billed demand. Once these dominant peaks have been reduced, additional EVs can only target progressively smaller remaining peaks, resulting in diminishing marginal demand-charge reductions.

On the cost side, however, the monthly cost to the fleet operator continues to grow with fleet size, as illustrated in Figure~\ref{cost_ev_operator_vs_evs} for the Tiered setup. Under the assumptions considered in this study, labor cost is the dominant contributor to this increase. Consequently, net savings eventually plateau and may decline when the additional operating costs exceed the incremental revenue from further peak shaving. While diminishing marginal demand-charge reduction is a general consequence of peak-based billing, the fleet size that maximizes net savings depends on the specific cost assumptions, tariff structure, load characteristics, and charger capacity. Therefore, the best-performing fleet sizes identified in this study should be interpreted as case-specific results rather than universal conclusions.

 \begin{figure}[htbp]
    \centering
    \begin{subfigure}{0.48\textwidth}
        \centering
        \includegraphics[width=\linewidth]{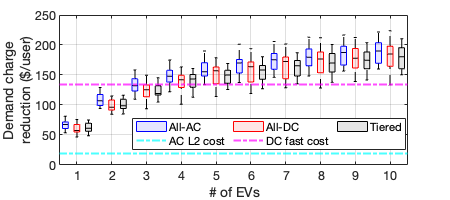}
        \caption*{(a) B-10 user during winter months.}
    \end{subfigure}
    \begin{subfigure}{0.48\textwidth}
       \centering
        \includegraphics[width=\linewidth]{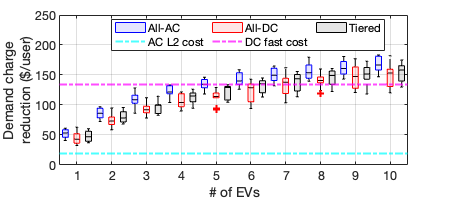}
        \caption*{(b) B-10 user during summer months.}
    \end{subfigure}
    \hfill
        \begin{subfigure}{0.48\textwidth}
        \centering
        \includegraphics[width=\linewidth]{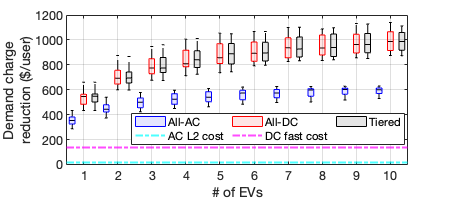}
        \caption*{(c) B-19 user during winter months.}
    \end{subfigure}
    \begin{subfigure}{0.48\textwidth}
    \centering
        \includegraphics[width=\linewidth]{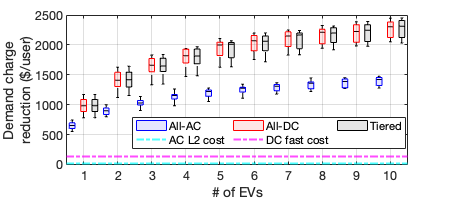}
        \caption*{(d) B-19 user during summer months.}
    \end{subfigure}
    \caption{Demand charge reduction with number of EVs under different infrastructure configurations by season and user type.}
    \label{dcr_mult_ev_all_setups}
\end{figure}

\begin{figure}[h]
\centering
\includegraphics[width=0.45\textwidth]{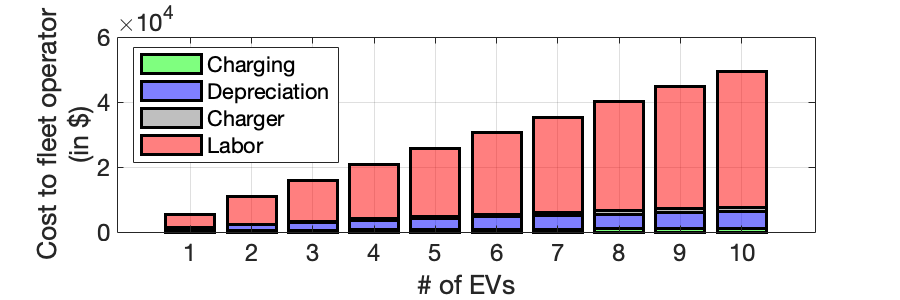}
\caption{Breakdown of the monthly cost to fleet operator under the Tiered setup.}
\label{cost_ev_operator_vs_evs}
\end{figure}

To further examine the utilization of additional fleet capacity, we evaluate the average peak-demand reduction achieved per service under different fleet sizes, seasons, and charging-infrastructure configurations, as shown in Figure~\ref{peak_dem_red_mult_evs}. The results indicate that the average peak-demand reduction per service generally decreases as fleet size increases. This behavior is expected because the first few EVs are assigned to the most critical peak-demand intervals, yielding the largest reductions in customer peak demand. As additional EVs are deployed, they are increasingly assigned to lower-priority intervals with smaller peak-reduction potential, resulting in diminishing marginal benefits per service. Therefore, the observed decrease in average peak-demand reduction per service reflects the diminishing value of additional fleet capacity rather than reduced operational effectiveness. These findings are consistent with the demand-charge-reduction results reported in Figure~\ref{dcr_mult_ev_all_setups} and further confirm that the proposed framework prioritizes the intervals that contribute most significantly to customer demand charges.

\begin{figure}[htbp]
    \centering
    \begin{subfigure}{0.48\textwidth}
        \centering
        \includegraphics[width=\linewidth]{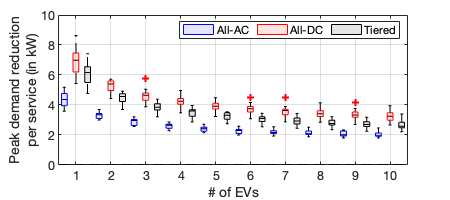}
        \caption{Peak demand reduction during winter.}
        \label{peak_dem_red_mult_evs_winter}
    \end{subfigure}
    \hfill
    \begin{subfigure}{0.48\textwidth}
        \centering
        \includegraphics[width=\linewidth]{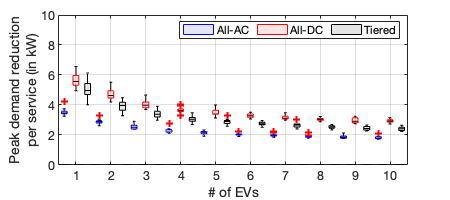}
        \caption{Peak demand reduction during summer.}
        \label{peak_dem_red_mult_evs_summer}
    \end{subfigure}
    \caption{Peak demand reduction with number of EVs under different infrastructure configurations by season.}
    \label{peak_dem_red_mult_evs}
\end{figure}

\begin{remark}
The net savings achieved by the proposed framework depend strongly on the underlying tariff structure, since demand charge formulations, seasonal pricing, and billing mechanisms directly influence the economic value of peak reduction. Consequently, identical EV dispatch strategies may yield different levels of savings under different utility tariffs. To illustrate this sensitivity, additional evaluations under the Puget Sound Energy (PSE) tariff are provided in Appendix~\ref{app:PSE_results}.
\end{remark}

\subsubsection*{Sensitivity to uncertainty}

Figure~\ref{savings_vs_uncertainty} illustrates the effect of load forecast uncertainty on net savings for a representative winter and summer month under the tiered setup with three EVs. Offline dispatch schedules are computed using imperfect load forecasts obtained by perturbing the true load with random errors whose magnitude is expressed as a percentage of the average load, and the resulting schedules are then applied to the true load to compute net savings. As forecast uncertainty increases, schedules computed using imperfect forecasts become progressively misaligned with the realized load, leading to a monotonic decline in net savings. The decline is more pronounced in summer due to higher demand charges and sharper peaks, while winter savings decline more gradually. Beyond moderate uncertainty levels (5\% in winter and 10\% in summer), net savings vanish, indicating reduced effectiveness of demand charge reduction under higher forecast uncertainty. If driver labor costs are eliminated, as would be the case with autonomous EVs, the break-even uncertainty threshold extends to approximately 15\% in winter and 19\% in summer, placing the business model within uncertainty levels commonly encountered in practical operational settings.

 \begin{figure}[h]
\centering
\includegraphics[width=0.45\textwidth]{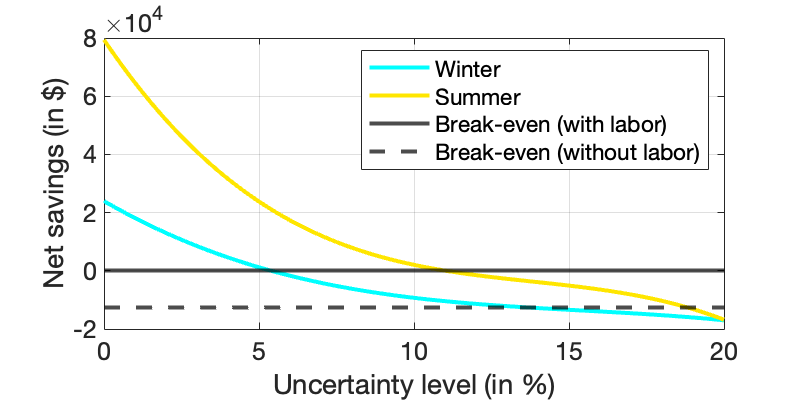}
\caption{Net savings with uncertainty.}
\label{savings_vs_uncertainty}
\end{figure}

There could be other operational uncertainties that influence the performance of the proposed framework. For example, variations in travel time due to traffic conditions can affect EV arrival schedules and delay peak-shaving services during critical intervals. Similarly, EV availability may vary due to driver scheduling uncertainties, vehicle maintenance, accidents, or other unexpected operational conditions, which can reduce the number of EVs available for dispatch at a given time. Charger outages or communication failures may also temporarily limit service capability and reduce operational flexibility. Although the current work provides an initial analysis of load forecast uncertainty, these additional uncertainties could be incorporated in future studies through stochastic or scenario-based modeling approaches to evaluate their impact on dispatch decisions, demand charge reduction, and overall economic performance.

\section{Conclusion}\label{sec:CON}
 
This paper develops a high-fidelity fleet management framework to quantify the true economic value of shared mobile energy storage for demand charge reduction. Unlike idealized models, our approach explicitly accounts for the transit energy consumption, labor costs, and battery degradation unders a unified framework. We formulated the dispatch problem as a MILP and developed a marginal-value-based heuristic algorithm to efficiently solve it under these complex physical constraints. Using empirical load profiles of user in San Francisco, California served by PG\&E, we show that this business model is economically attractive, particularly in summer months compared to winter. The Tiered setup consistently yields higher net savings than both the All-AC and All-DC setups across both seasons. During winter months, the best-performing fleet size for the Tiered setup is three EVs, achieving monthly net savings of approximately \$30,000. During summer months, the best-performing fleet size increases to six EVs, yielding monthly net savings approaching \$100,000. Large users achieve substantially higher savings than smaller ones, though the marginal benefit of additional EVs diminishes with scale. Future work will extend this framework to account for uncertainty in user load profiles, travel-time variations, EV availability, and charger reliability.

\bibliographystyle{IEEEtran}
\bibliography{references}

\appendices

\section{Proof of Lemma~1}
\label{app:proof_lemma1}

Fix EV $j$ and its dispatch matrix $\mathbf{M}_{j}$; together with the users' service decisions $\mathbf{u}^{\star}$ this fixes the constants $\rho_{j,t}:=\tfrac{1}{\delta_{j}^{\mathrm{tr}}}\big\|(x_{j,t},y_{j,t})-(x_{j,t-1},y_{j,t-1})\big\|_{1}\ge 0$ (the required transit energy, i.e.\ the right-hand side of \eqref{ev:transit_ener}), $e^{\mathrm{dis}}_{j,t}\ge 0$, and $m_{j,t}\in\{0,1\}$. With \eqref{ev:transit_ener} relaxed to $e^{\mathrm{tr}}_{j,t}\ge \rho_{j,t}$, we prove that every optimal solution $(\mathbf{e}^{\mathrm{ch}*}_{j},\mathbf{b}^{*}_{j},\mathbf{e}^{\mathrm{tr}*}_{j})$ of the charging subproblem \eqref{ev:charging_cost} satisfies $e^{\mathrm{tr}*}_{j,t}=\rho_{j,t}$ for all $t$.

\smallskip
\noindent\textit{Reduction to one day.}
By the end-of-day reset \eqref{eq:daily_cyc}, the state of charge is pinned to $b^{\mathrm{EOD}}_{j}$ at the end-of-day interval of each day, so consecutive days share that pinned value and decouple. Let $\mathcal{T}_{d}=\{t^{0}_{d},\dots,t^{\mathrm{E}}_{d}\}$ denote the intervals of day $d$, with $t^{\mathrm{E}}_{d}\in\mathcal{T}^{\mathrm{EOD}}$ and $t^{0}_{d}-1$ the end-of-day interval of the previous day, so that both endpoints of day $d$ are fixed,
\begin{equation}\label{eq:lem_pin}
b^{*}_{j,t^{0}_{d}-1}=b^{*}_{j,t^{\mathrm{E}}_{d}}=b^{\mathrm{EOD}}_{j}
\end{equation}
(for the first day, $t^{0}_{d}-1$ carries the fixed initial state of charge; the argument uses only that both endpoints are fixed). A perturbation confined to $\mathcal{T}_{d}$ that preserves $b_{j,t^{\mathrm{E}}_{d}}$ leaves every other day, every other EV, and $\mathbf{M}_{j},\mathbf{e}^{\mathrm{dis}}_{j}$ unchanged; we therefore argue within a fixed day $d$.

\smallskip
\noindent\textit{Daily energy balance.}
Summing the dynamics \eqref{eq:batt_dyn} over $\tau\in\mathcal{T}_{d}$ telescopes the state-of-charge terms; with \eqref{eq:lem_pin},
\begin{equation}\label{eq:lem_balance}
\eta_{j}^{\mathrm{ch}}\sum_{\tau\in\mathcal{T}_{d}}e^{\mathrm{ch}*}_{j,\tau}
=\frac{1}{\eta_{j}^{\mathrm{dis}}}\sum_{\tau\in\mathcal{T}_{d}}\big(e^{\mathrm{dis}}_{j,\tau}+e^{\mathrm{tr}*}_{j,\tau}\big),
\end{equation}
so the total daily charging is an increasing function of the total daily transit energy.

\smallskip
\noindent\textit{Slack hypothesis.}
Suppose, for contradiction, that the relaxed constraint is slack at some $t\in\mathcal{T}_{d}$:
\begin{equation}\label{eq:lem_slack}
\varepsilon_{0}:=e^{\mathrm{tr}*}_{j,t}-\rho_{j,t}>0,\qquad\text{whence } e^{\mathrm{tr}*}_{j,t}>0.
\end{equation}
Since $e^{\mathrm{tr}*}_{j,t}>0$, EV $j$ changes location between $t-1$ and $t$ and is therefore not at the station in both intervals, so $m_{j,t-1}m_{j,t}=0$ and the charging gate \eqref{ev:char} forces
\begin{equation}\label{eq:lem_nocharge}
e^{\mathrm{ch}*}_{j,t}=0 .
\end{equation}
(This uses only $e^{\mathrm{tr}*}_{j,t}>0$; the value $\rho_{j,t}$ is irrelevant, so the case $\rho_{j,t}=0$ needs no separate treatment.) Hence $t$ is not a charging interval, and by \eqref{eq:lem_slack} the right-hand side of \eqref{eq:lem_balance} is strictly positive; as $\eta_{j}^{\mathrm{ch}}>0$, the set of charging intervals $\mathcal{C}_{d}:=\{\tau\in\mathcal{T}_{d}:e^{\mathrm{ch}*}_{j,\tau}>0\}$ is nonempty with $t\notin\mathcal{C}_{d}$.

\smallskip
\noindent\textit{Two compensating perturbations.}
Define the nearest charging intervals on either side of $t$ (each possibly absent),
\[
r:=\min\{\tau\in\mathcal{C}_{d}:\tau>t\},\qquad s:=\max\{\tau\in\mathcal{C}_{d}:\tau<t\}.
\]
For a step $\epsilon\in(0,\varepsilon_{0}]$, set $\tilde e^{\mathrm{tr}}_{j,t}:=e^{\mathrm{tr}*}_{j,t}-\epsilon\ (\ge\rho_{j,t}\ge0)$, keep all other transit values, and compensate at a single interval $\sigma\in\{r,s\}$ by $\tilde e^{\mathrm{ch}}_{j,\sigma}:=e^{\mathrm{ch}*}_{j,\sigma}-\epsilon/(\eta_{j}^{\mathrm{ch}}\eta_{j}^{\mathrm{dis}})$, regenerating $\tilde{\mathbf{b}}_{j}$ from \eqref{eq:batt_dyn}. Reducing the drain at $t$ by $\epsilon/\eta_{j}^{\mathrm{dis}}$ and the charge injection at $\sigma$ by $\eta_{j}^{\mathrm{ch}}\cdot\epsilon/(\eta_{j}^{\mathrm{ch}}\eta_{j}^{\mathrm{dis}})=\epsilon/\eta_{j}^{\mathrm{dis}}$, the state-of-charge deviation $w_{j,\tau}:=\tilde b_{j,\tau}-b^{*}_{j,\tau}$ is
\begin{equation}\label{eq:lem_dev}
w_{j,\tau}=\frac{\epsilon}{\eta_{j}^{\mathrm{dis}}}\big(\mathbf{1}[\tau\ge t]-\mathbf{1}[\tau\ge\sigma]\big),\qquad \tau\in\mathcal{T}_{d}.
\end{equation}
Thus $w_{j,\tau}=0$ outside the window strictly between $t$ and $\sigma$; in particular $w_{j,t^{\mathrm{E}}_{d}}=0$, so the reset \eqref{eq:daily_cyc} and the pin \eqref{eq:lem_pin} are preserved automatically and \eqref{eq:batt_dyn} holds by construction. As the perturbation only lowers $e^{\mathrm{ch}}_{j,\sigma}$ and leaves its cap untouched, the gate \eqref{ev:char} holds provided $\tilde e^{\mathrm{ch}}_{j,\sigma}\ge0$, i.e.\ $\epsilon\le\eta_{j}^{\mathrm{ch}}\eta_{j}^{\mathrm{dis}}e^{\mathrm{ch}*}_{j,\sigma}$, which holds for small $\epsilon>0$ since $e^{\mathrm{ch}*}_{j,\sigma}>0$. By \eqref{eq:lem_dev}, feasibility then reduces to the box \eqref{eq:batt_lev}, and the two options carry disjoint risks:
\begin{itemize}
\item[(A)] $\sigma=r$: $w_{j,\tau}=+\epsilon/\eta_{j}^{\mathrm{dis}}>0$ on $[t,r-1]$ and $0$ elsewhere; only the \emph{upper} bound can fail, and only if $b^{*}_{j,\tau}=\overline{b}_{j}$ for some $\tau\in[t,r-1]$.
\item[(B)] $\sigma=s$: $w_{j,\tau}=-\epsilon/\eta_{j}^{\mathrm{dis}}<0$ on $[s,t-1]$ and $0$ elsewhere; only the \emph{lower} bound can fail, and only if $b^{*}_{j,\tau}=\underline{b}_{j}$ for some $\tau\in[s,t-1]$.
\end{itemize}
If the relevant one-sided bound is strictly inactive throughout the window, \eqref{eq:batt_lev} holds for all sufficiently small $\epsilon>0$.

\smallskip
\noindent\textit{At least one option is feasible.}
Since $\mathcal{C}_{d}\neq\varnothing$ and $t\notin\mathcal{C}_{d}$, at least one of $r,s$ exists.

\emph{Both exist.} As $r,s$ are the nearest charging intervals straddling $t$, no charging occurs strictly between them, so by \eqref{eq:batt_dyn} $b^{*}_{j,\tau}-b^{*}_{j,\tau-1}=-(1/\eta_{j}^{\mathrm{dis}})(e^{\mathrm{dis}}_{j,\tau}+e^{\mathrm{tr}*}_{j,\tau})\le0$ for $\tau\in\{s+1,\dots,r-1\}$; hence $b^{*}_{j,\cdot}$ is non-increasing on $\{s,\dots,r-1\}$, which contains $[s,t-1]$ followed by $[t,r-1]$. If both options failed, then $b^{*}_{j,k_{1}}=\overline{b}_{j}$ for some $k_{1}\in[t,r-1]$ and $b^{*}_{j,k_{2}}=\underline{b}_{j}$ for some $k_{2}\in[s,t-1]$; but $k_{2}<t\le k_{1}$ lie in the non-increasing stretch, so $\underline{b}_{j}=b^{*}_{j,k_{2}}\ge b^{*}_{j,k_{1}}=\overline{b}_{j}$, contradicting $\underline{b}_{j}<\overline{b}_{j}$. Thus one window is strictly interior and yields a feasible $\epsilon>0$.

\emph{Only $s$ exists.} Then no charging occurs on $(s,t^{\mathrm{E}}_{d}]$, so for $k\in[s,t-1]$, summing \eqref{eq:batt_dyn} over $(k,t^{\mathrm{E}}_{d}]$ and using \eqref{eq:lem_pin},
\[
b^{*}_{j,k}=b^{\mathrm{EOD}}_{j}+\frac{1}{\eta_{j}^{\mathrm{dis}}}\sum_{\tau=k+1}^{t^{\mathrm{E}}_{d}}\big(e^{\mathrm{dis}}_{j,\tau}+e^{\mathrm{tr}*}_{j,\tau}\big)\ge b^{\mathrm{EOD}}_{j}+\frac{e^{\mathrm{tr}*}_{j,t}}{\eta_{j}^{\mathrm{dis}}}>b^{\mathrm{EOD}}_{j}\ge\underline{b}_{j},
\]
because the strictly positive drain $e^{\mathrm{tr}*}_{j,t}$ occurs at $t>k$. The lower bound is strictly inactive on $[s,t-1]$, so (B) is feasible for small $\epsilon>0$.

\emph{Only $r$ exists.} Then no charging occurs on $[t^{0}_{d},r-1]$, so for $k\in[t,r-1]$, summing \eqref{eq:batt_dyn} over $[t^{0}_{d},k]$ and using \eqref{eq:lem_pin},
\[
b^{*}_{j,k}=b^{\mathrm{EOD}}_{j}-\frac{1}{\eta_{j}^{\mathrm{dis}}}\sum_{\tau=t^{0}_{d}}^{k}\big(e^{\mathrm{dis}}_{j,\tau}+e^{\mathrm{tr}*}_{j,\tau}\big)\le b^{\mathrm{EOD}}_{j}-\frac{e^{\mathrm{tr}*}_{j,t}}{\eta_{j}^{\mathrm{dis}}}<b^{\mathrm{EOD}}_{j}\le\overline{b}_{j},
\]
because the drain at $t\le k$ is strictly positive. The upper bound is strictly inactive on $[t,r-1]$, so (A) is feasible for small $\epsilon>0$.

\smallskip
\noindent\textit{Strict improvement.}
For a feasible option and a sufficiently small $\epsilon>0$, the perturbation lowers only $e^{\mathrm{tr}}_{j,t}$ (by $\epsilon$) and $e^{\mathrm{ch}}_{j,\sigma}$ (by $\epsilon/(\eta_{j}^{\mathrm{ch}}\eta_{j}^{\mathrm{dis}})$), with $\mathbf{e}^{\mathrm{dis}}_{j}$ fixed. In \eqref{Dep_cost} the depreciation term $\gamma_{j}\mathbf{1}^{\intercal}(\mathbf{e}^{\mathrm{dis}}_{j}+\mathbf{e}^{\mathrm{tr}}_{j})$ decreases by exactly $\gamma_{j}\epsilon>0$. In \eqref{ev:charging_cost}, $\tilde{\mathbf{e}}^{\mathrm{ch}}_{j}\le\mathbf{e}^{\mathrm{ch}*}_{j}$ componentwise, so the energy charge $\pmb{\pi}_{0}^{\intercal}\sum_{j}\mathbf{e}^{\mathrm{ch}}_{j}$ weakly decreases ($\pmb{\pi}_{0}\ge0$), and lowering EV $j$'s charging at the single interval $\sigma$ makes every aggregate $\sum_{j}e^{\mathrm{ch}}_{j,\tau}$ non-increasing, so each peak $\max_{\tau\in\mathcal{T}^{\theta}}\sum_{j}e^{\mathrm{ch}}_{j,\tau}$ and hence the demand charge weakly decreases. All other objective terms---in particular the users' demand charges---are unchanged. The total objective therefore strictly decreases (by at least $\gamma_{j}\epsilon>0$), contradicting optimality. Hence no slack \eqref{eq:lem_slack} can occur, i.e.\ $e^{\mathrm{tr}*}_{j,t}=\rho_{j,t}$ for all $t$: the relaxation of \eqref{ev:transit_ener} is binding at every optimum.

\section{Computational Performance Evaluation}

\subsection{Comparison with MILP Bounds}

To assess the computational performance and solution quality of the proposed algorithm relative to the exact MILP formulation, we compare their performance on representative multi-EV instances. Figure~\ref{fig:runtime_convergence} summarizes the convergence behavior of both approaches. Figures~\ref{fig:runtime_convergence}(a) and \ref{fig:runtime_convergence}(c) show the evolution of the objective value over time. The red line denotes the incumbent objective value reported by Gurobi, the black line denotes the lower bound obtained from the MILP relaxation, and the blue line denotes the solution obtained by the proposed algorithm. Figures~\ref{fig:runtime_convergence}(b) and \ref{fig:runtime_convergence}(d) show the corresponding MILP optimality gap trajectories. The upper and lower bounds are generated by Gurobi during its presolve, root-node relaxation, and branch-and-bound procedures, enabling the progress of the MILP solution process to be monitored over time.

The results show that the proposed algorithm obtains a high-quality solution within practical run-times (less than 15 minutes), while the MILP continues refining its incumbent solution and bounds throughout the five-hour time limit. Moreover, a nonzero MIP optimality gap remains at the end of the simulation horizon, indicating that the MILP has not fully converged for these instances. The proposed algorithm's solution remains close to the best-known MILP bounds, providing empirical evidence of its effectiveness for solving large-scale multi-EV scheduling problems within practical computational times.

\begin{figure*}[tbh]
    \centering

    \begin{subfigure}[b]{0.48\textwidth}
        \centering
        \includegraphics[width=\linewidth]{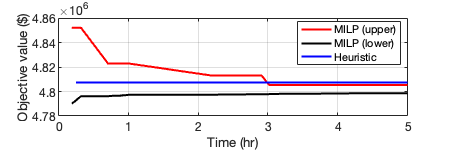}
        \caption{Winter: Upper and lower bound evolution.}
        \label{fig:winter_bounds}
    \end{subfigure}
    \hfill
    \begin{subfigure}[b]{0.48\textwidth}
        \centering
        \includegraphics[width=\linewidth]{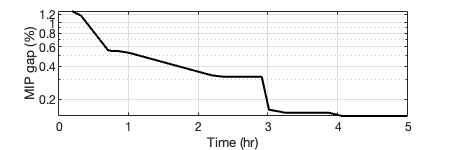}
        \caption{Winter: MILP gap evolution.}
        \label{fig:winter_gap}
    \end{subfigure}

    \vspace{0.4cm}

    \begin{subfigure}[b]{0.48\textwidth}
        \centering
        \includegraphics[width=\linewidth]{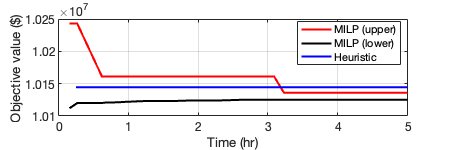}
        \caption{Summer: Upper and lower bound evolution.}
        \label{fig:summer_bounds}
    \end{subfigure}
    \hfill
    \begin{subfigure}[b]{0.48\textwidth}
        \centering
        \includegraphics[width=\linewidth]{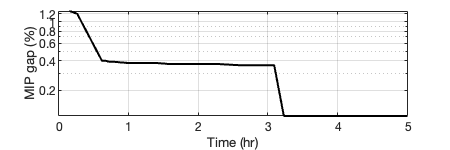}
        \caption{Summer: MILP gap evolution.}
        \label{fig:summer_gap}
    \end{subfigure}

    \caption{Evolution of optimization performance metrics for representative winter and summer months. Subfigures (a) and (c) show the evolution of the incumbent upper bound and best lower bound over runtime, while subfigures (b) and (d) illustrate the corresponding MILP gap convergence behavior.}
    
    \label{fig:runtime_convergence}
\end{figure*}

\subsection{Algorithm performance with fleet size}

To evaluate scalability with respect to fleet size, we varied the number of EVs from 1 to 20 while keeping all other parameters unchanged. The resulting runtimes are reported in Figure~\ref{fig:runtimes_evs}, where the blue, red, and black curves correspond to the All-AC, All-DC, and Tiered charging configurations, respectively. The figure shows that runtime increases approximately linearly with the fleet size for all charging configurations. Table~\ref{tab:runtime_summary_fleet_size} reports the corresponding mean runtimes and standard deviations across charging configurations. The standard deviations remain small relative to the means, indicating limited runtime variability across the months. The All-AC setup exhibits the lowest runtime across all fleet sizes. For fleet sizes up to approximately 16 EVs, the Tiered and All-DC setups exhibit similar runtimes. However, for larger fleet sizes, the Tiered setup becomes computationally less demanding than the All-DC setup. This behavior is attributable to the mixed charger deployment in the Tiered setup, which results in fewer feasible service opportunities than the All-DC setup and consequently a smaller search space for the heuristic algorithm.

\begin{figure}[htbp]
\centering
\includegraphics[width=0.5\linewidth]{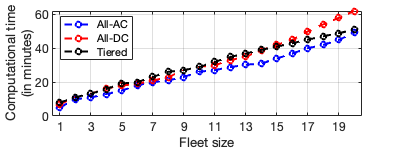}
\caption{Computation time of the proposed algorithm with fleet size.}
\label{fig:runtimes_evs}
\end{figure}

Regarding solution quality, direct comparison with the optimal MILP solution is impractical for large instances. Even for the 136-user case with only two EVs, the MILP did not converge within the imposed five-hour time limit and retained a nonzero optimality gap. This computational difficulty highlights the practical need for efficient heuristic approaches when considering realistic large-scale deployments.

\begin{table*}[htbp]
\centering
\caption{Computational times in minutes of the proposed algorithm with fleet size (mean $\pm$ standard deviation).}
\label{tab:runtime_summary_fleet_size}

\resizebox{\textwidth}{!}{%
\begin{tabular}{l*{10}{c}}
\hline
Fleet size & 1 & 2 & 3 & 4 & 5 & 6 & 7 & 8 & 9 & 10 \\
\hline
All-AC
& $5.02 \pm 0.25$ & $9.93 \pm 0.56$
& $10.93 \pm 0.67$ & $12.99 \pm 0.74$
& $14.97 \pm 0.90$ & $17.91 \pm 1.05$
& $20.01 \pm 1.03$ & $20.98 \pm 1.09$
& $23.28 \pm 1.36$ & $25.82 \pm 1.36$ \\
All-DC
& $6.96 \pm 0.37$ & $10.94 \pm 0.59$
& $12.92 \pm 0.62$ & $16.07 \pm 0.91$
& $17.79 \pm 0.91$ & $19.01 \pm 1.10$
& $21.08 \pm 1.13$ & $22.93 \pm 1.50$
& $27.10 \pm 1.79$ & $28.86 \pm 1.82$ \\
Tiered
& $8.02 \pm 0.41$ & $10.99 \pm 0.58$
& $12.99 \pm 0.85$ & $16.07 \pm 0.93$
& $18.90 \pm 1.00$ & $19.76 \pm 1.04$
& $23.04 \pm 1.23$ & $26.05 \pm 1.67$
& $27.07 \pm 1.62$ & $28.79 \pm 1.64$ \\
\hline
\end{tabular}%
}

\par\vspace{2mm}

\resizebox{\textwidth}{!}{%
\begin{tabular}{l*{10}{c}}
\hline
Fleet size & 11 & 12 & 13 & 14 & 15 & 16 & 17 & 18 & 19 & 20 \\
\hline
All-AC
& $27.09 \pm 1.52$ & $28.77 \pm 1.69$
& $29.96 \pm 2.10$ & $31.25 \pm 1.93$
& $33.71 \pm 2.35$ & $36.93 \pm 2.35$
& $40.20 \pm 3.04$ & $42.35 \pm 2.55$
& $45.69 \pm 2.81$ & $49.12 \pm 3.52$ \\
All-DC
& $29.75 \pm 1.76$ & $32.45 \pm 1.94$
& $35.10 \pm 2.49$ & $38.80 \pm 2.56$
& $41.73 \pm 2.86$ & $45.41 \pm 2.63$
& $50.00 \pm 3.53$ & $54.08 \pm 3.17$
& $58.52 \pm 4.68$ & $61.96 \pm 4.22$ \\
Tiered
& $31.93 \pm 2.05$ & $35.32 \pm 2.39$
& $37.25 \pm 2.54$ & $38.93 \pm 1.99$
& $40.62 \pm 2.58$ & $43.21 \pm 2.97$
& $44.58 \pm 2.37$ & $46.90 \pm 2.80$
& $49.21 \pm 3.46$ & $51.39 \pm 2.89$ \\
\hline
\end{tabular}%
}
\end{table*}

\subsection{Algorithm performance with number of buildings}

The original smart meter dataset consisted of 730 business buildings. After data preprocessing, users whose load profiles contained more than 10\% missing measurements were removed, resulting in a dataset of 470 buildings. Among these, 136 users were originally subject to demand charges under PG\&E's B-10 or B-19 tariffs and were therefore included in the primary case study. The remaining users exhibited peak demand levels below 75 kW and were not eligible for these demand-charge tariffs.

To investigate scalability over a larger population, we generated additional simulation instances by proportionally scaling the load profiles of these remaining users such that their peak demands fell within the B-10 or B-19 customer categories. This scaling was performed solely for computational benchmarking purposes and should not be interpreted as representing realistic customer demand-charge savings. Consequently, the resulting savings values are not used for economic analysis; rather, these instances serve exclusively to evaluate the computational performance and scalability of the proposed algorithm.

Specifically, we considered user populations of 50, 100, 136, 150, 200, 250, 300, 350, 400, 450, and 470 users, where the 136-user case corresponds to the primary case study presented in the manuscript. All results correspond to a fleet size of three EVs and are evaluated for representative summer and winter months. The resulting runtime performance for different numbers of users is reported in Figure~\ref{fig:runtimes_users}, where the blue, red, and black curves correspond to the All-AC, All-DC, and Tiered charging configurations, respectively. 

\begin{figure}[htbp]
    \centering
    \includegraphics[width=0.5\linewidth]{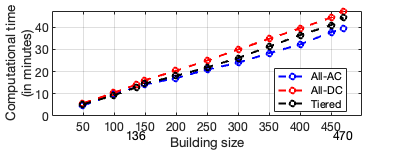}
    \caption{Computational time of the proposed algorithm with number of buildings.}
    \label{fig:runtimes_users}
\end{figure}

The figure shows that runtime increases with the number of users for all charging configurations, exhibiting an approximately linear trend. Table~\ref{tab:runtime_buildings} reports the corresponding mean runtimes and standard deviations across charging configurations. The standard deviations remain small relative to the means, indicating limited runtime variability across the months. Furthermore, the All-AC setup consistently exhibits the lowest runtime, followed by the Tiered setup and then the All-DC setup. Consistent with fleet size study, this behavior arises because lower charging capacities reduce the number of feasible service intervals that satisfy the energy requirements, thereby decreasing the number of candidate assignments evaluated by the algorithm. In contrast, higher charging capacities increase the number of feasible service opportunities, resulting in a larger search space and consequently longer runtimes.

Regarding solution quality, direct comparison with the optimal MILP solution becomes computationally intractable for large instances. In particular, for instances containing more than 100 users with a fleet size of three EVs, the MILP formulation did not converge within the imposed five-hour time limit and retained a non-zero optimality gap. Consequently, exact optimality-gap comparisons are not feasible for the larger user populations considered in this scalability study. This observation further demonstrates the practical need for computationally efficient heuristic methods when addressing realistic large-scale deployments.

\begin{table*}[htbp]
\centering
\caption{Computational times of the proposed algorithm with number of buildings (mean $\pm$ standard deviation, in minutes).}
\label{tab:runtime_buildings}
\resizebox{\textwidth}{!}{%
\begin{tabular}{l*{11}{c}}
\hline
Number of buildings & 50 & 100 & 136 & 150 & 200 & 250 & 300 & 350 & 400 & 450 & 470 \\
\hline
All-AC & $4.79 \pm 0.98$ & $9.57 \pm 0.90$ & $12.98 \pm 1.88$ & $14.19 \pm 1.48$ & $17.03 \pm 2.22$ & $20.85 \pm 2.24$ & $23.91 \pm 2.04$ & $28.24 \pm 2.55$ & $32.14 \pm 3.38$ & $37.67 \pm 3.79$ & $39.16 \pm 2.04$ \\
All-DC & $5.74 \pm 1.28$ & $10.45 \pm 0.71$ & $14.16 \pm 1.66$ & $16.13 \pm 1.01$ & $20.36 \pm 2.11$ & $24.94 \pm 1.35$ & $29.90 \pm 2.57$ & $34.65 \pm 0.59$ & $39.40 \pm 2.07$ & $44.23 \pm 1.78$ & $47.21 \pm 3.03$ \\
Tiered & $5.24 \pm 0.53$ & $9.29 \pm 1.43$ & $12.89 \pm 1.29$ & $14.59 \pm 1.60$ & $18.04 \pm 1.72$ & $21.84 \pm 3.07$ & $26.05 \pm 3.17$ & $31.22 \pm 3.66$ & $36.30 \pm 2.56$ & $40.84 \pm 2.37$ & $44.29 \pm 2.84$ \\
\hline
\end{tabular}%
}
\end{table*}

\section{Evaluation under Puget Sound Energy Tariffs}\label{app:PSE_results}
This appendix presents additional results using tariffs from Puget Sound Energy to evaluate the proposed framework under an alternative commercial electricity pricing structure. The same 136 users from the main case study are considered here and reclassified under PSE Schedule 25 or Schedule 26 based on their mean monthly peak demand. Schedule 25 applies to users with peak demand between 50–350 kW (71 users), while Schedule 26 applies to users above 350 kW (65 users). Both tariffs include seasonal demand rates, with higher rates during winter months (October–March). Under Schedule 25, the demand rates are \$11.41/kW in summer and \$17.10/kW in winter, while under Schedule 26 the corresponding rates are \$14.42/kW and \$21.63/kW, respectively. Unlike the PG\&E tariffs, the PSE tariffs do not include TOU demand rates. For the charging station, PSE Schedule 25 is considered, with seasonal energy rates of \$0.1153/kWh during summer months and \$0.1243/kWh during winter months.

Figure~\ref{pse_net_savings_mult_evs} shows the monthly net savings under the PSE tariffs for different fleet sizes and infrastructure configurations across winter and summer seasons. Overall, the achievable savings are significantly lower than those obtained under the PG\&E tariffs. This reduction is primarily due to the absence of TOU demand charges in the PSE tariff structure, which limits the economic benefit through targeted peak shaving.

\begin{figure}[htbp]
    \centering
    \begin{subfigure}{0.48\textwidth}
        \centering
        \includegraphics[width=\linewidth]{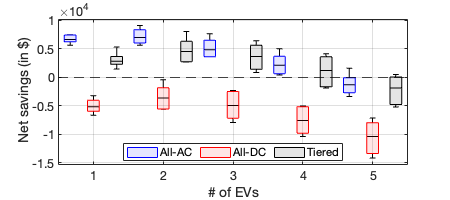}
        \caption{Net savings during winter months.}
        \label{pse_net_savings_mult_evs_winter}
    \end{subfigure}
    \hfill
    \begin{subfigure}{0.48\textwidth}
        \centering
        \includegraphics[width=\linewidth]{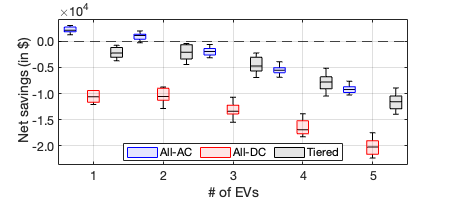}
        \caption{Net savings during summer months.}
        \label{pse_net_savings_mult_evs_summer}
    \end{subfigure}
    \caption{Net savings with number of EVs under different infrastructure configurations by season.}
    \label{pse_net_savings_mult_evs}
\end{figure}

Net savings are consistently higher during winter months because the PSE tariffs impose larger demand charges during October--March. In contrast, during summer months, nearly all configurations yield negative net savings, indicating that the achievable demand charge reductions are insufficient to offset charging, labor, transit, and infrastructure-related costs. The only profitable summer configuration corresponds to the all-AC setup with a single EV.

During winter months, the all-AC setup with two EVs achieves the highest monthly net savings of approximately \$7,500. The all-DC configuration remains economically unfavorable across all fleet sizes and seasons due to its higher infrastructure cost. Meanwhile, the tiered setup does not consistently outperform the all-AC setup, suggesting that the additional fast-charging capability does not always provide sufficient economic benefit under the PSE tariff structure.

These results indicate that the economic viability of the proposed business model is highly dependent on the tariff design. Under tariffs with relatively moderate demand charges and without TOU demand pricing, profitable operation may only be achievable during high-demand seasons and for limited fleet sizes.

To further interpret the net savings behavior, Figure~\ref{pse_dcr_mult_ev_all_setups} show the demand charge reduction achieved for Schedule 25 and Schedule 26 users across different seasons, fleet sizes, and infrastructure configurations. In general, the demand charge reduction initially increases with the number of EVs but gradually saturates as the dominant peaks are reduced. For Schedule 25 users, the achieved reductions during both summer and winter months are generally insufficient to offset the higher infrastructure costs associated with DC fast chargers, although they remain adequate for AC charging configurations. In contrast, Schedule 26 users achieve substantially larger demand charge reductions, making DC fast charging economically feasible in several cases.

 \begin{figure}[htbp]
    \centering
    \begin{subfigure}{0.48\textwidth}
        \centering
        \includegraphics[width=\linewidth]{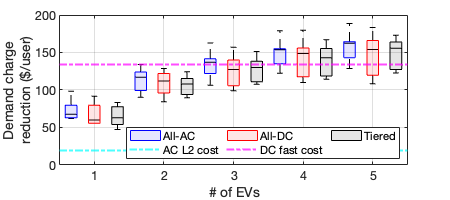}
        \caption*{(a) Schedule-25 user during winter months.}
    \end{subfigure}
    \begin{subfigure}{0.48\textwidth}
       \centering
        \includegraphics[width=\linewidth]{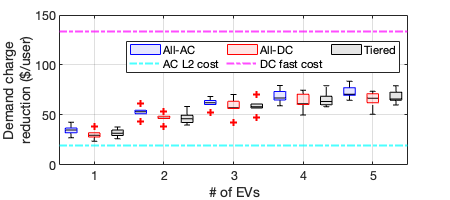}
        \caption*{(b) Schedule-25 user during summer months.}
    \end{subfigure}
    \hfill
        \begin{subfigure}{0.48\textwidth}
        \centering
        \includegraphics[width=\linewidth]{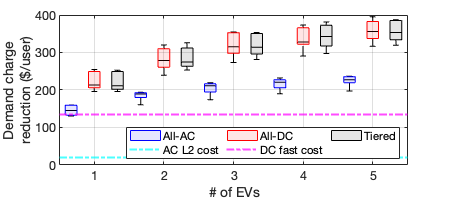}
        \caption*{(c) Schedule-26 user during winter months.}
    \end{subfigure}
    \begin{subfigure}{0.48\textwidth}
    \centering
        \includegraphics[width=\linewidth]{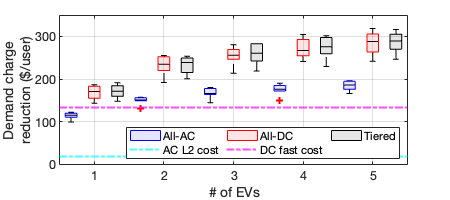}
        \caption*{(d) Schedule-26 user during summer months.}
    \end{subfigure}
    \caption{Demand charge reduction with number of EVs under different infrastructure configurations by season and user type.}
    \label{pse_dcr_mult_ev_all_setups}
\end{figure}

Figure~\ref{pse_cost_ev_operator_vs_evs} further presents the monthly cost breakdown for the all-AC configuration, showing that labor cost constitutes the dominant operational expense for the fleet operator.

\begin{figure}[htbp]
\centering
\includegraphics[width=0.45\textwidth]{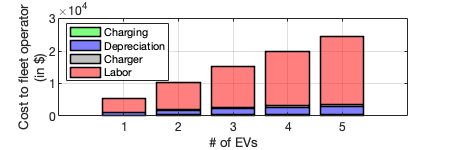}
\caption{Breakdown of the monthly cost to fleet operator under the All-AC setup.}
\label{pse_cost_ev_operator_vs_evs}
\end{figure}

\end{document}